\documentclass[12pt]{article}
\usepackage[authoryear,round]{natbib}
\usepackage{amssymb}
\usepackage{graphicx}
\usepackage{latexsym}
\usepackage{url}
\usepackage[left=1in,top=1in,right=1in,bottom=1in]{geometry}

\usepackage{xcolor}
\definecolor{Bleu}{RGB}{0,0,204}
\definecolor{Violet}{RGB}{102,0,204}
\definecolor{Rouge}{RGB}{204,0,0}
\definecolor{Highlight}{RGB}{251,0,0}
\usepackage[nokeyprefix]{refstyle}
\usepackage{varioref}

\usepackage{amsmath}
\usepackage{amsthm,thmtools}
\usepackage{booktabs}
\usepackage{color}
\usepackage{tikz}
\usepackage{rotating}

\usepackage[compact]{titlesec}
\titlespacing*{\section}
{0pt}{6pt plus 1pt minus 1pt}{0pt plus 0pt minus 0pt}
\titlespacing*{\subsection}
{0pt}{6pt plus 1pt minus 1pt}{1pt plus 1pt minus 1pt}

\usepackage{enumitem}
\setenumerate{itemsep=0pt,topsep=6pt}
\makeatletter
\newcommand{\mylabel}[2]{
	\addtocounter{\@listctr}{-1}%
    \refstepcounter{\@listctr}%
	#2\def\@currentlabel{#2}\label{#1}}
\makeatother

\usepackage{algorithm}
\usepackage{algorithmicx}
\usepackage[noend]{algpseudocode}
\usepackage{setspace}
\usepackage{etoolbox} 
\AtBeginEnvironment{algorithm}{\setstretch{1.35}} 

\usepackage{mathrsfs} 

\usepackage{dsfont} 

\usepackage{accents}

\usetikzlibrary{shapes}
\usetikzlibrary{arrows}
\definecolor{darkblue}{rgb}{0,0.4,0.9}
\definecolor{gray10}{rgb}{0.1,0.1,0.1}
\definecolor{gray20}{rgb}{0.2,0.2,0.2}
\definecolor{gray30}{rgb}{0.3,0.3,0.3}
\definecolor{gray40}{rgb}{0.4,0.4,0.4}
\definecolor{gray60}{rgb}{0.6,0.6,0.6}
\definecolor{gray80}{rgb}{0.8,0.8,0.8}
\definecolor{gray90}{rgb}{0.9,0.9,.9}
\definecolor{gray95}{rgb}{0.95,0.95,.95}
\definecolor{gray96}{rgb}{0.96,0.96,.96}
\definecolor{lgreen} {RGB}{180,210,100}
\definecolor{dblue}  {RGB}{20,66,129}
\definecolor{ddblue} {RGB}{11,36,69}
\definecolor{lred}   {RGB}{220,0,0}
\definecolor{nred}   {RGB}{224,0,0}
\definecolor{norange}{RGB}{230,120,20}
\definecolor{nyellow}{RGB}{255,221,0}
\definecolor{ngreen} {RGB}{98,158,31}
\definecolor{dgreen} {RGB}{78,138,21}
\definecolor{nblue}  {RGB}{28,130,185}
\definecolor{jblue}  {RGB}{20,50,100}
\definecolor{nnyellow}{RGB}{235,200,0}
\definecolor{purple}{RGB}{150, 0, 120}
\definecolor{sgGreen} {RGB}{20, 180, 50}
\definecolor{revised}{rgb}{0,0,0.9}

\newtheoremstyle{mystyle}
  {6pt plus 0pt minus 0pt} 
  {6pt plus 0pt minus 0pt} 
  {} 
  {} 
  {\sc} 
  {.} 
  {.5em} 
  {} 
\theoremstyle{mystyle}
\newtheorem{definition}{Definition}

\newtheorem{theorem}{Theorem}

\newtheorem{lemma}[theorem]{Lemma}

\declaretheorem[name=Remark,qed={$\Box$}]{remark} 

\newcommand{\openr}{\hbox{${\rm I\kern-.2em R}$}}
\newcommand{\openn}{\hbox{${\rm I\kern-.2em N}$}}
\newcommand{\logit}{\operatorname{logit}}

\newcommand{\Rem}{\operatorname{Rem}}

\newcommand{\norm}[1]{\left\lVert#1\right\rVert}

\newcommand{\argmin}{\operatorname{argmin}}
\newcommand{\E}{\mathbb{E}}

\newcommand{\Ind}{\mathds{1}}
\newcommand{\IF}{\operatorname{IF}}
\newcommand{\ER}{\mathcal{E}}
\newcommand{\Normal}{\operatorname{Normal}}
\newcommand{\Bern}{\operatorname{Bernoulli}}

\newcommand{\phR}{Z}

\makeatletter
\newlength{\trianglerightwidth}
\algnewcommand{\LineCommentCont}[1]{\Statex \hskip\ALG@thistlm%
  \parbox[t]{\dimexpr\linewidth-\ALG@thistlm}{\hangindent=\trianglerightwidth \hangafter=1 \strut$\triangleright$ #1\strut}}
\makeatother

\allowdisplaybreaks

\usepackage[normalem]{ulem}
\usepackage{longtable}
\usepackage{authblk}

\title{Sequential Double Robustness in Right-Censored Longitudinal Models \\[1em]
{\normalsize Version 2}\\\vspace{-1em}
{\normalsize Version 1: May 6, 2017}}
\author[1,2]{Alexander R.  Luedtke}

\affil[1]{\footnotesize   Vaccine  and   Infectious  Disease   Division,  Fred
  Hutchinson Cancer Research Center, USA}
\affil[2]{\footnotesize   Public Health Sciences  Division,  Fred
  Hutchinson Cancer Research Center, USA} 

\author[3]{Oleg Sofrygin}
\affil[3]{\footnotesize   Division of Biostatistics, University of California, Berkeley, USA} 

\author[3]{Mark J. van der Laan}

\author[4,1]{Marco Carone}
\affil[4]{\footnotesize   Department of Biostatistics, University of Washington, USA}

\begin{document}
\maketitle
\begin{abstract}\singlespacing
Consider estimating the G-formula for the counterfactual mean outcome under a given treatment regime in a
longitudinal study. \cite{Bang&Robins05}
provided an estimator for this parameter based on a sequential regression formulation. This approach is doubly robust in that
it is consistent if either all outcome regressions or all treatment mechanisms are consistently
estimated. 
We define a stronger notion of double robustness, termed sequential double robustness, for estimators of the longitudinal G-formula. The definition emerges naturally from a more general definition of sequential double robustness for outcome regression estimators. An outcome regression estimator is sequentially doubly robust (SDR) if, at each subsequent time point, either the outcome regression or the treatment mechanism is consistently estimated. This form of robustness is exactly what one would anticipate is attainable by studying the remainder term of a first-order expansion of the G-formula parameter. We introduce a novel SDR estimator, whose development involves a novel infinite-dimensional extension of targeted minimum loss-based estimation. These new developments have the potential to dramatically improve both the robustness of estimators of the marginal G-formula and the rate of convergence of outcome regression estimators, with improved rates translating into dramatically improved finite sample behavior.\vspace{.2in}
\end{abstract}


\section{Introduction}\label{sec:intro}
Consider a longitudinal study, where for each individual in the study we have observed time-varying covariates and treatment indicator, and we also observe an outcome at the end of the study. Suppose the goal is to estimate the G-formula that, under the consistency and sequential randomization assumptions \citep{Robins1986}, identifies the mean of the counterfactual outcome had each participant received treatment at every time point. Estimating the end-of-study mean outcome in a study with dropout can be equivalently described in this manner. Throughout, we refer to the probability of receiving treatment at time $t$, conditional on the observed past, as the time $t$ treatment mechanism, and the G-formula identified with the mean counterfactual outcome, conditional on the observed past before time $t$ treatment, under receiving treatment at each time point at or beyond $t$ as the time $t$ outcome regression.

In the last several decades, there has been extensive work on estimating causal parameters from right-censored longitudinal data structures. \cite{Tsiatisetal2011} and \cite{Rotnitzkyetal12} provide comprehensive reviews -- here, we provide only an abbreviated version. Early approaches for G-formula estimation include inverse probability weighted methods \citep[e.g.,][]{Robins1993,Robinsetal2000} and structural nested mean/G-estimation methods \citep[e.g.,][]{Robins89,Robins94}. These approaches are consistent when the treatment mechanisms or the outcome regressions, respectively, are consistently estimated, and they have normal limits when parametric models have been specified for these components of the observed data distribution. More recently, there has been extensive development of doubly robust (DR) methods, which are consistent if either the treatment mechanism at each time point or the outcome regression at each time point is consistently estimated, and have normal limits under additional conditions. \cite{Robins1999b} introduced a sequential methodology for estimating this quantity, which represents an extension of the single time point methodology in \cite{Scharfstein&Rotnitzky&Robins99}. \cite{Bang&Robins05} introduced a simplification of the approach based on sequential estimation of regressions conditional on the past history. \Citet{vanderLaan&Gruber12} extended this procedure to allow for data-adaptive estimation. \cite{Rotnitzkyetal12} adapted it to allow for oracle-type model selection that achieves the optimal efficiency within a prespecified class when the treatment mechanism is correctly specified by a parametric model. \cite{Tsiatisetal2011} describe an approach to exploit a correctly specified parametric model for the treatment/missingness mechanism, though using estimating equations rather than a sequential regression procedure. \cite{Seaman&Copas2009} describe a DR generalized estimating equation methodology for longitudinal data structures.

In this work, we describe a different form of robustness, which we term sequential double robustness. An estimator of the G-formula parameter is sequentially doubly robust (SDR) if it is consistent provided that, at each time point, either the outcome regression or the treatment mechanism is consistently estimated. This property is stronger than the traditional definition of double robustness, which either requires consistency of all of the outcome regression estimators or all of the treatment mechanism estimators. We note that sequential double robustness represents a special case of $2^K$ multiple robustness, where $K$ is the number of time points \citep{Vansteelandtetal2007}. We find the name sequential double robustness appealing because it stressed the sense in which this robustness is an extension of existing double robustness: sequentially across time points. In the case $K=2$, \cite{Tchetgen2009} exhibited an SDR estimator of the G-formula based on parametric modeling of the outcome regressions and treatment mechanisms. They also outlined an extension to the case $K>2$. In this work, we consider methods allowing for data-adaptive estimators, formulate the SDR estimators of arbitrary outcome regressions, present a plug-in estimator of the outcome regression (rather than an estimating equation-based approach), and rigorously prove the SDR property for these data adaptive estimators -- our results are thus more general than the important early work of \cite{Tchetgen2009}. We show that an existing estimator achieves this property and propose a new estimator that is also guaranteed to respect known bounds on the outcome. The development of this estimator involves translating ideas from targeted minimum loss-based estimation (TMLE) into estimation of infinite-dimensional parameters that are not pathwise differentiable and for which square root of sample size convergence rates are not typically possible. This extension is distinct from the recent work of \cite{vdL&Gruber2016}, which describes a TMLE for infinite-dimensional parameters for which each component is pathwise differentiable. We instead draw inspiration from the recent work of \cite{Kennedyetal2016}, in which an infinite-dimensional targeting step based on locally linear regression is implemented in a continuous point treatment setting.

Other authors have also recently become interested in SDR estimators. Contemporaneously to the initial technical report of this work \citep{Luedtkeetal2017}, \citeauthor{Rotnitzkyetal2017} posted a technical report with related results \citep{Rotnitzkyetal2017} that were obtained independently. Both of these works study the behavior of flexible algorithms for estimating the G-formula. Although both \cite{Rotnitzkyetal2017} and the present work study SDR estimators, they have a significantly different focus. \cite{Rotnitzkyetal2017} focus on an estimating equation approach. While we discuss this approach through the concept of ``doubly robust unbiased transformations'' in Section~\ref{sec:existing}, our main focus is on a TMLE extension, which yields a plug-in estimator of the outcome regressions and the corresponding G-formula. 
Furthermore, in the rigorous study of \cite{Rotnitzkyetal2017}, the allowed flexible learning algorithms must be linear operators. This work, on the other hand, does not require linearity, but instead focuses on empirical risk minimizers over generic Donsker classes. The recent work of \cite{Molinaetal2017} studies multiple robustness in general models in which the likelihood factorizes as a product of variation-independent quantities, and refers to the SDR estimator of the G-formula presented in \cite{Tchetgen2009}. Unlike the setting of our work, \cite{Molinaetal2017} focuses on settings where certain components of the likelihood are indexed by finite-dimensional nuisance parameters.

This paper is organized as follows. In Section~\ref{sec:obj}, the parameter(s) of interest are introduced and  the notion of sequential double robustness is formalized. A variation-independent, but more restrictive, formulation of sequential double robustness is presented in Appendix~\ref{app:varindep}. An analysis of the SDR properties of certain existing data-adaptive outcome regression estimators is given in Section~\ref{sec:existing}. A general template for constructing an SDR estimator is given in Section~\ref{sec:generaltemplate}. In Section~\ref{sec:newproc}, our new SDR procedure, which we refer to as an infinite-dimensional targeted minimum loss-based estimator (iTMLE), is presented. Formal properties of the empirical risk minimization (ERM) variant of our procedure are given in Section~\ref{sec:erm}. In practice, ERMs may be prone to overfitting when used with our procedure -- in Appendix~\ref{app:cv}, we present a variant that relies on cross-validation to mitigate this problem. We recommend this variant for use in practice. Though more notationally burdensome, proofs for the cross-validated procedure and the basic ERM approach are nearly identical and are thus omitted. Simulation results are presented in Section~\ref{sec:sim}, and a discussion is provided in Section~\ref{sec:disc}. 
All proofs not provided in the main text can be found in Appendix~\ref{app:proofs}.

\section{Notation, definitions and objective}\label{sec:obj}
\noindent\textit{Parameter(s) of interest.} Let $O=(L_0,A_0,L_1,\ldots,A_{K},L_{K+1})$ be the observed longitudinal  data unit, where indices denote time (e.g., visit number, time periods elapsed), $L_t$ is a vector of covariates recorded at time $t$, $A_t$ is the treatment node value (or indicator of being under study) at time $t$, and $L_{K+1}=Y$, the outcome at time-point $K+1$, a fixed time of interest. For each $t$, denote by $\bar{H}_t=(L_0,A_0,L_1,\ldots,L_t)$ the history recorded right before $A_t$ is determined. In particular, $O=\bar{H}_{K+1}$. For simplicity of exposition, suppose that outcome $L_{K+1}$ is bounded in the unit interval. Let $\Psi:x\mapsto (1+e^{-x})^{-1}$ be the expit function and $\Psi^{-1}:x\mapsto\log[x/(1-x)]$ the logit function.

Suppose we observe an i.i.d. sample $O_1,O_2,\ldots,O_n$ drawn from a distribution $P$ belonging to some model $\mathcal{M}$. 
Throughout, we will refer to an arbitrary element in $\mathcal{M}$ by $P'$. 
We will use $\E$ and $\E'$ to denote expectations under $P$ and $P'$, respectively. For each $t=1,\ldots,K$, we define the treatment mechanism as $\pi_t: \bar{h}_t\mapsto P(A_t=1|\bar{H}_t=\bar{h}_t)$, and use $\pi_{t,i}$ and $\hat{\pi}_{t,i}$ as shorthand notation for $\pi_t(\bar{H}_{t,i})$ and $\hat{\pi}_t(\bar{H}_{t,i})$, respectively, with $\hat{\pi}_t$ a given estimator of $\pi_t$. Throughout, we make the strong positivity assumption: there exists a $\delta>0$ so that $P\{\pi_t(\bar{H}_t)>\delta\}=1$ for each $t=1,\ldots,K$.

We use many recursions over $t=K,\ldots,0$ in this work. When $t=0$ this requires some conventions: $\sum_{s=1}^{0}\ldots = 0$ (sums are zero); $\prod_{s=1}^{0} \ldots = 1$ (products are one); $\bar{H}_{0}=\emptyset$ (the time $0$ covariate is empty); $A_{0}=1$ and $\pi_t(\bar{h}_0)= 1$ (the time $0$ intervention is always $1$); and, for $f : \bar{\mathcal{H}}_{0}\rightarrow\mathbb{R}$, $f(\bar{H}_{0})=f$ (functions applied to $\emptyset$ can also be written as constants).

For a history vector $\bar{h}_{K+1}$, set $Q_{K+1}:\bar{h}_{K+1}\mapsto \ell_{K+1}$. For $t=K,\ldots,1$, recursively define
\begin{align*}
Q_t:\bar{h}_t\mapsto \E[Q_{t+1}\left(\bar{H}_{t+1}\right)|\bar{h}_t,A_t=1],
\end{align*}
and define $Q_0\equiv \E[Q_1(\bar{H}_1)]$. For ease of notation, we write $Q_{t,i}\equiv Q_t\left(\bar{H}_{t,i}\right)$. For an estimate $\hat{Q}_t$, we similarly write $\hat{Q}_{t,i}$. For $t\ge 1$, our objective will be to estimate $Q_t$ as well as possible in terms of some user-specified criterion. We focus on the mean-squared error criterion in this work. For $t=0$, our objective will be to obtain a consistent estimate of $Q_t$ for which there exist reasonable conditions for its asymptotically linearity.

\vspace{6pt}\noindent\textit{(Sequential) double robustness.}
We will make use of the following non-technical conditions, defined for each $t=0,\ldots,K$.
\begin{enumerate}[leftmargin=*,widest={OR.$t$}]
	\item[\mylabel{it:or}{OR.$t$)}]  The functional form of the outcome regression at time $t$, i.e. $Q_t$, is correctly specified by the estimation procedure, or at least arbitrarily well approximated asymptotically.
	\item[\mylabel{it:tm}{TM.$t$)}]  The treatment mechanism at time point $t$, i.e. $\pi_t$, is consistently estimated.
\end{enumerate}
We use these conditions informally to discuss properties of existing estimator and our new estimator until Section~\ref{sec:newproc}, where we begin to present formal conditions for the validity of our estimator. Note that TM.$t$ requires consistent estimation, and OR.$t$ requires correct specification. This discrepancy occurs because we will use OR.$t$ as a part of a sufficient condition for consistent estimation of $Q_t$. 
For estimation of $Q_0$, double robustness is defined as follows \citep{vdL02,Bang&Robins05,Tsiatisetal2011}:
\begin{definition}\label{def:DRQm1}
A \textit{doubly robust} estimator of $Q_0$ is consistent if either (i) OR.$t$ holds for all $t=1,\ldots,K$ or (ii) TM.$t$ holds for all $t=1,\ldots,K$.
\end{definition}
These estimators are referred to as DR because there are two possibilities for obtaining consistent estimation. In this work, we define a more general form of robustness.
\begin{definition}\label{def:SDRQm1}
A \textit{sequentially doubly robust} estimator of $Q_0$ is an estimator that is consistent if, for each $t=1,\ldots,K$, either OR.$t$ or TM.$t$ holds.
\end{definition}
An SDR estimator is double robust, but the converse need not hold.  This estimator could be referred to as $2^K$-robust \citep{Vansteelandtetal2007} in the sense that there are $2^K$ ways that the estimation can satisfy one and only one of OR.$t$ and TM.$t$, $t=1,\ldots,K$, though of course certain of these possibilities may be less likely than others. 
Appendix~\ref{app:varindep} gives a variation-independent, but more restrictive, formulation of the OR.$t$ conditions. We also define sequential double robustness for the outcome regressions.
\begin{definition}\label{def:SDRQk}
A \textit{sequentially doubly robust} estimator of $Q_t$ is an estimator that is consistent if OR.$t$ holds and, for each $s>t$, either OR.$s$ or TM.$s$ holds.
\end{definition}
Because OR.$t$ is a triviality when $t=0$ (the functional form is correctly specified by a constant), the definition of sequential double robustness given specifically for $Q_0$ is a special case of this definition. 
The objective of this work will be to present sequentially doubly robust estimators of $Q_t$, $t=0,\ldots,K$. We will also outline arguments showing that sequentially doubly robust estimators should have faster rates of convergence than sequential regression estimators, and therefore in finite samples are expected to give more precise estimates of the outcome regressions, e.g. the baseline-covariate-conditional mean counterfactual risk under treatment at all time points. Our estimator of $Q_0$ will also be efficient among all regular and asymptotically linear estimators under some additional conditions.

\section{Detailed overview of related estimation procedures}\label{sec:existing}
In the introduction, we gave a broad overview of the literature for estimating mean outcomes from monotonely coarsened data structures. We now describe several semi- or nonparametric methods from this literature. One could also study parametric methods, incorporating basis function transformations of the covariates of increasing dimension to allow for increasingly flexible estimation of the outcome regressions. We do not consider such approaches here.

First, we describe a method that uses DR unbiased transformations of the data, i.e. distribution dependent pseudo-outcomes with conditional expectation equal to the parameter of interest given correct specification of an outcome regression or treatment mechanism. Variants of these unbiased transformations were given in \cite{Rubin&vanderLaan2007}, which represent a DR extension of the unbiased transformations presented earlier in the literature \citep[see, e.g.,][]{Buckley&James1979,Koul&Susarla&vanRyzin81}. 
We describe an SDR implementation of this unbiased transform estimator, also discussed in \cite{Rotnitzkyetal2017}, and we also discuss its shortcomings. We then describe inverse probability weighted (IPW) loss functions as presented in \cite{vanderLaan&Dudoit03}. 
Finally, we discuss sequential regression procedures in the vein of \cite{Bang&Robins05}.

\vspace{6pt}\noindent\textit{Doubly robust unbiased transformations.} An unbiased transform for $Q_t$ is a distribution-dependent mapping $\Gamma_t^P : \bar{\mathcal{H}}_{K+1}\rightarrow\mathbb{R}$ such that $\Gamma_t^P(\bar{H}_{K+1})$ has mean $Q_t(\bar{H}_t)$ when $\bar{H}_{K+1}$ is drawn from the conditional distribution of $P$ given that $(\bar{H}_t,A_t)=(\bar{h}_t,1)$. Early work on these transformations used imputation-based approaches \citep[e.g.,][]{Buckley&James1979}, whose consistency relies on consistently estimating $Q_{s}$, $s>t$, or on IPW approaches \citep[e.g.,][]{Koul&Susarla&vanRyzin81}, whose consistency relies on consistently estimating $\pi_t$, $s>t$. \cite{Rubin&vanderLaan2007} presented an AIPW unbiased transformation, of which a special case was presented in \cite{Rotnitzkyetal2006}. For our problem, one could estimate $Q_t$ using the DR transform
\begin{align*}
\Gamma_{t,i}&\equiv \sum_{s=t+1}^K \left(\prod_{r=t+1}^{s} \frac{A_{r,i}}{\pi_{r,i}}\right)\left\{Q_{s+1,i} - Q_{s,i}\right\} + Q_{t+1,i},
\end{align*}
regressing $\Gamma_{t,i}$ against $\bar{H}_{t,i}$ for all subjects $i$ with $A_{t,i}=1$. In practice, the transformation $\widehat{\Gamma}_{t,i}$ 
is used, where each instance of $Q_{s}$ and $\pi_{s}$, $s>t$, is replaced by estimates $\hat{Q}_{s,i}$, $\hat{\pi}_{s,i}$. Consider the procedure that regresses $\widehat{\Gamma}_{t,i}$ against $\bar{H}_{t,i}$ for all $i$ such that $A_{t,i}=1$. This procedure is DR in the sense that it is consistent if OR.$t$ holds and either (i) OR.$s$ holds for all $s>t$ or (ii) TM.$s$ holds for all $s>t$. If applied iteratively as in Algorithm~\ref{alg:druntrnas}, i.e. if $\hat{Q}_K$ is estimated via this procedure, then $\hat{Q}_{K-1}$ using as initial estimate $\hat{Q}_K$, then $\hat{Q}_{K-2}$ using as initial estimates $\hat{Q}_{K-1}$ and $\hat{Q}_K$, etc., then this procedure is SDR in the sense of Definition~\ref{def:SDRQk} \citep{Luedtkeetal2017,Rotnitzkyetal2017}.

The advantage of this procedure is that it is easy to implement: indeed, once one has the transformation $\widehat{\Gamma}_t$, one can simply plug the observations $\{\widehat{\Gamma}_{t,i},\,\bar{H}_{t,i} : i\textnormal{ satisfies }A_{t,i}=1\}$ into their preferred regression tool. The disadvantage of this procedure is that the transformations do not necessarily obey the bounds of the original outcome: indeed, $L_{K+1}$ may be bounded in $[0,1]$, but $\widehat{\Gamma}_{t,i}$ can be very large in absolute value if $\prod_{s>t}\hat{\pi}_{s,i}$ gets close to zero. While one could theoretically constrain the estimated regression function to respect the $[0,1]$ bounds, few existing regression software packages allow for such constraints on the model. The stability of such a procedure has also not been evaluated when there are near-positivity violations, i.e. $\prod_{t=1}^{K} \hat{\pi}_{t}(\bar{H}_t)$ is near zero. We evaluate this procedure in Section~\ref{sec:sim}.

\vspace{2pt}{\centering
[Algorithm~\ref{alg:druntrnas} about here.]\par
}\vspace{2pt}

\vspace{6pt}\noindent\textit{(A)IPW loss functions.} An alternative approach uses the IPW loss function of \cite{vanderLaan&Dudoit03}. Before describing this approach, we introduce the notion of a loss function. Suppose that one wishes to estimate a feature $f_0$ of a distribution $\nu$. For example, $\nu$ may be the distribution of a predictor-outcome pair $Z\equiv (X,Y)$, and $f_0$ may be the conditional expectation function $x\mapsto \E_{\nu}[Y|x]$. For each $f$, let $z\mapsto \mathscr{L}(z;f)$ denote a real-valued function. We call $\mathscr{L}$ a loss function if $f_0=\argmin_f \E_{\nu}[\mathscr{L}(Z;f)]$ over an appropriate index set for $f$. The quantity $\E_{\nu}[\mathscr{L}(Z;f)]$ is referred to as the risk of $f$. Examples of loss functions for the conditional mean $f_0$ include the squared-error loss $(z;f)\mapsto (y-f(x))^2$ and, for $z$ bounded in $[0,1]$, the cross-entropy loss $(z;f)\mapsto -[y\log f(x) + (1-y)\log \{1-f(x)\}]$.

We now present the IPW loss. For simplicity, we focus on the IPW squared-error loss, given by $\mathscr{L}_t(\bar{h}_{K+1};Q_{t}')\equiv \left(\prod_{s=t+1}^K \frac{a_{s}}{\pi_{s}(\bar{h}_{s})}\right)\{\ell_{K+1} - Q_{t}'(\bar{h}_t)\}^2$. The standard change of measure argument associated with IPW estimators shows that $Q_t$ equals the argmin over $Q_{t}'$ of $\E\left[\mathscr{L}_t(\bar{H}_{K+1};Q_{t}')\middle|A_t=1\right]$, where the $Q_{t}'$-specific expectation is referred to as the risk of $Q_{t}'$. This suggests that one can estimate $Q_t$ by minimizing the empirical risk conditional on $A_t=1$, i.e. the empirical mean of $\mathscr{L}_t(\bar{H}_{K+1,i};Q_{t}')$ among subjects with $A_{t,i}=1$. 
In practice we may not know $\pi_{s}$, so we replace them with estimates and denote the corresponding loss by $\widehat{\mathscr{L}}_{s}$. 
This procedure yields a consistent estimate of $Q_t$ if each $\hat{\pi}_{s}$ is consistently estimated and the regression is correctly specified. There also exists an AIPW version of this loss function \citep{vanderLaan&Dudoit03}, though this loss does not appear to easily yield an SDR procedure \citep[see][]{Luedtkeetal2017}.

\vspace{6pt}\noindent\textit{Sequential regression.} \cite{Bang&Robins05} proposed a procedure that takes advantage of the recursive definition of these $Q_t$ functions. They aimed to estimate $Q_0$. An instance of their procedure for binary outcomes is displayed in Algorithm~\ref{alg:bangrobins}. In short, it first correctly specifies $\hat{Q}_{K+1}(\bar{H}_{K+1})\equiv L_{K+1}$. Now, iteratively from $t=K$ to $t=0$, it uses observations $i$ satisfying $\prod_{s=1}^t A_{s,i}=1$ to regress $\hat{Q}_{t+1,i}$ against $\bar{H}_{t,i}$, using a parametric fit and the logit link function. Each parametric fit includes a linear term with covariate $\frac{1}{\prod_{s=1}^{t}\hat{\pi}_{s,i}}$. These linear terms were added to ensure that $\hat{Q}_0$ solves the efficient estimating equation, i.e. that
\begin{align*}
\sum_{i=1}^n\sum_{t=0}^{K} \left[\left\{\prod_{s=1}^{t}\frac{A_{s,i}}{\hat{\pi}_{s,i}}\right\}\left\{\hat{Q}_{t+1,i}-\hat{Q}_{t,i}\right\}\right] = 0.
\end{align*}
In particular, the fitting of $\hat{\epsilon}_t$ ensures that each $t$-specific term in the above sum equals zero. \Citet{vanderLaan&Gruber12} extended this procedure to allow $\hat{Q}_t$ to be estimated data adaptively. They refer to their estimator as a longitudinal TMLE (LTMLE). 

While both the \cite{Bang&Robins05} and \cite{vanderLaan&Gruber12} procedures are DR for $\hat{Q}_0$, neither is SDR for $\hat{Q}_0$. Furthermore, neither is SDR for $Q_t$, $t\ge 1$, in the sense of Definition~\ref{def:SDRQk}. In particular, when $t\ge 1$ consistent estimation of $Q_t$ from these procedures relies on OR.$s$, $s\ge t$. These procedures can be consistent when only some outcome regressions and some treatment mechanisms are consistently estimated \citep{Molinaetal2017,Luedtkeetal2017}. In particular, if there exists some $t$ such that OR.$s$ holds for all $s\ge t$ and TM.$s$ holds for all $s<t$, then these estimators will be consistent for $Q_0$.

\vspace{2pt}{\centering
[Algorithm~\ref{alg:bangrobins} about here.]\par
}\vspace{2pt}

\section{General template for achieving sequential double robustness}\label{sec:generaltemplate}
We now give a general template for achieving sequential double robustness. This template hinges on a straightforward induction argument, where we show that achieving an SDR estimator at time $t+1$ yields an SDR estimator at time $t$. In this section, we let $\hat{Q}_t$ and $\hat{\pi}_t$ respectively denote generic estimates of the outcome regression and treatment mechanism at time $t$. We will often make use of the following strong positivity assumption on our treatment mechanism estimates. Though we introduced this condition earlier, we name it here for clarity. This condition can be enforced in the estimation procedure via truncation.
\begin{enumerate}[leftmargin=*,widest={SP.$t$}]
	\item[\mylabel{it:sp}{SP.$t$)}] There exists a $\delta>0$ such that, for each $s>t$, $P\{\hat{\pi}_{s}(\bar{H}_{s})>\delta\}=0$.
\end{enumerate}
To ease notation, when $Q_{s}$ or $\pi_{s}$, or an estimate thereof, fall within an expectation, we often omit the dependence on $\bar{H}_{s}$. For a real-valued function $\bar{h}_{s}\mapsto f(\bar{h}_{s})$, we denote the $L^2(P)$ norm by $\norm{f}=\E[f(\bar{H}_{s})^2]^{1/2}$. We define the following useful objects:
\begin{align}
&\textnormal{D}_{s}^t(\hat{Q}_{s+1},\hat{Q}_{s})(\bar{h}_t)\equiv -\E\left[\left\{\frac{A_t}{\pi_t}\prod_{r=t+1}^{s} \frac{A_{r}}{\hat{\pi}_{r}}\right\}\left\{\hat{Q}_{s+1}-\hat{Q}_{s}\right\}\middle|\bar{h}_t\right],\;s\ge t \nonumber\\
&\Rem_{s}^t(\hat{Q}_{s})(\bar{h}_t)\equiv \E\left[\left\{\frac{A_t}{\pi_t}\prod_{r=t+1}^{s-1}\frac{A_{r}}{\hat{\pi}_{r}}\right\}\left(1-\frac{\pi_{s}}{\hat{\pi}_{s}}\right)\left\{\hat{Q}_{s}-Q_{s}\right\}\middle|\bar{h}_t\right],\;s>t. \label{eq:RemDdef}
\end{align}
Note that, for each $(\hat{Q}_{s+1},\hat{Q}_{s})$, $\textnormal{D}_{s}^t(\hat{Q}_{s+1},\hat{Q}_{s})$ is a function mapping a $\bar{h}_t$ to the real line, and similarly for $\Rem_{s}^t(\hat{Q}_{s})$. In the remainder of this section we study a particular estimator $(\hat{Q}_{s+1},\hat{Q}_{s})$, and so use the simpler notation $\textnormal{D}_{s}^t\equiv \textnormal{D}_{s}^t(\hat{Q}_{s+1},\hat{Q}_{s})$ and $\Rem_{s}^t\equiv \Rem_{s}^t(\hat{Q}_{s})$. We also define $\textnormal{D}^t(\bar{h}_t)\equiv \sum_{s=t}^{K} \textnormal{D}_{s}^t(\bar{h}_t)$.
\begin{lemma}[First-order expansion of $Q_t$] \label{lem:firstordexp}
If SP.$t$ holds, then, for $P$-almost all $\bar{h}_t\in\bar{\mathcal{H}}_t$,
\begin{align}
\hat{Q}_t(\bar{h}_t)-Q_t(\bar{h}_t)&= \textnormal{D}^t(\bar{h}_t)+\sum_{s=t}^{K} \Rem_{s}^t(\bar{h}_t). \label{eq:ptwiseident}
\end{align}
\end{lemma}
By the triangle inequality, Cauchy-Schwarz, and the positivity assumption on $\hat{\pi}_t$,
\begin{align}
\norm{\hat{Q}_t - Q_t}&\le \norm{\textnormal{D}^t} + \sum_{s=t+1}^K \norm{\Rem_{s}^t}\le \norm{\textnormal{D}^t} + \sum_{s=t+1}^K O\left(\norm{\hat{\pi}_{s} - \pi_{s}}\norm{\hat{Q}_{s} - Q_{s}}\right). \label{eq:CS}
\end{align}
A simple induction argument shows that
\begin{align}
\norm{\hat{Q}_t - Q_t}&\le \norm{\textnormal{D}^t} + \sum_{s=t+1}^K O\left(\norm{\hat{\pi}_{s} - \pi_{s}}\norm{\textnormal{D}^{s}}\right).\label{eq:sdrtypedecomp}
\end{align}
The above teaches us how to obtain an SDR estimator, provided one is willing to replace OR.$t$, OR.$s$ by FO.$t$, FO.$s$ (``correct first-order behavior at time $t,s$''), where
\begin{enumerate}[leftmargin=*,widest={FO.$t$}]
	\item[\mylabel{it:fo}{FO.$t$)}] $\norm{\textnormal{D}^t}$ converges to zero in probability.
\end{enumerate}
The following result is an immediate consequence of (\ref{eq:sdrtypedecomp}).
\begin{theorem}[Achieving an SDR Estimator]\label{thm:achievingsdr}
Fix $t$ and suppose that SP.$t$ holds with probability approaching one. If FO.$t$ and, at each time $s>t$, either TM.$s$ or FO.$s$, then $\hat{Q}_t\rightarrow Q_t$, i.e. $\norm{\hat{Q}_t- Q_t}=o_P(1)$.
\end{theorem}
In Remark~\ref{rem:sdrunbiasednice}, we sketch an argument showing that FO.$t$ and OR.$t$ are equivalent for the (S)DR unbiased transformation approach. An induction argument also shows that, for the upcoming iTMLE approach, FO.$t$ and OR.$t$ are equivalent if, at each future time point $s$, either OR.$s$ or TM.$s$ holds. Hence, the use of FO.$t$, FO.$s$ in the above theorem rather than OR.$t$, OR.$s$ should not detract from its interest.

\begin{remark}\label{rem:sdrunbiasednice}
We now connect the (S)DR unbiased transformation approach to the first-order expansion (\ref{eq:ptwiseident}). For $P$ almost all $\bar{h}_t$, $\E\left[\widehat{\Gamma}_t\middle|A_t=1,\bar{h}_t\right] - \hat{Q}_t(\bar{h}_t) = - \textnormal{D}^t(\bar{h}_t)$. Thus, the $L^2(P)$ norm of $\E\left[\widehat{\Gamma}_t\middle|A_t=1,\bar{H}_t=\cdot\right] - \hat{Q}_t(\cdot)$ is $o_P(1)$ if and only if $\norm{\textnormal{D}^t}=o_P(1)$. Note that the objective of a regression of $\widehat{\Gamma}_{t,i}$ against $\bar{H}_{t,i}$ among all individuals with $A_{t,i}=1$ is to ensure $\E\left[\widehat{\Gamma}_t\middle|A_t=1,\bar{H}_t=\cdot\right]\approx \hat{Q}_t(\cdot)$, where this approximation can be made precise by using ``$\approx$'' to mean closeness in $L^2(P)$ norm. One could alternatively ensure closeness with respect to a different criterion by choosing a different loss function. Often closeness in one loss implies closeness in another loss (see, e.g., Theorem~\ref{thm:excessriskbd}).
\end{remark}

\begin{remark}
In (\ref{eq:CS}), we applied Cauchy-Schwarz to show that each $\norm{\Rem_{s}^t}$, $s>t$, is big-oh of $\norm{\hat{\pi}_{s} - \pi_{s}}\norm{\hat{Q}_{s} - Q_{s}}$. One can also obtain the bound
\begin{align*}
\norm{\Rem_{s}^t}&\lesssim \E\left[\E\left[\left\{\hat{Q}_{s}-Q_{s}\right\}^2\middle|\bar{H}_t\right]\E\left[\left(\hat{\pi}_{s}-\pi_{s}\right)^2\middle|\bar{H}_t\right]\right]^{1/2}.
\end{align*}
The left-hand side will converge if, for each $\bar{h}_t$, either $Q_{s}$ or $\pi_{s}$ is consistently estimated across all $\bar{h}_{s}$ that have time-$t$ history equal to $\bar{h}_t$. This is weaker than requiring that either $\pi_{s}$ or $Q_{s}$ is consistently estimated, since we only require that the union of $\bar{h}_t$ values on which each of these quantities are consistently estimated is equal to the support of $\bar{H}_t\sim P$.
\end{remark}

\begin{remark}
Let $\norm{f}_{\infty}$ denote the $P$ essential supremum norm and $\norm{\cdot}_1$ denote the $L^1(P)$ norm. In this case, we have that $\sum_{s=t+1}^K \norm{\Rem_{s}^t}_{\infty}\lesssim \norm{\hat{\pi}_{s}-\pi_{s}}_1 \norm{\hat{Q}_{s}^{t}-Q_{s}}_{\infty}$. This seems likely to be useful for constructing confidence bands for $Q_t$. In particular, this suggests that, under some conditions, $\norm{\hat{Q}_t - Q_t}\approx \norm{\textnormal{D}^t}$. Therefore, it generally suffices to develop a confidence band for $\norm{\textnormal{D}^t}$, which is a regression with the dimension of $\bar{H}_t$. If one uses the (S)DR unbiased transformation approach from Remark~\ref{rem:sdrunbiasednice} and implements the time $t$ regression using a kernel regression procedure, then one should be able to study a kernel-weighted empirical process to develop confidence bands. This will presumably give meaningful confidence bands when the dimension of $\bar{H}_t$ is not too large, e.g. when $t=1$ and the baseline covariate is low-dimensional. We will examine this in detail in future works.
\end{remark}

\section{Novel sequential regression procedure}\label{sec:newproc}
We now present a novel, SDR procedure for estimating $Q_{t^0}$, $t^0\ge 0$. We extend the univariate targeting step used in the procedure of \cite{vanderLaan&Gruber12} to infinite-dimensional targeting steps towards $Q_{t^0}$. We refer to this new procedure, presented in Algorithm~\ref{alg:sdr}, as iTMLE. For each $t$, we denote the estimate of $Q_t$ that is targeted towards all outcome regressions $r$ satisfying $s\le r\le t-1$ by $\hat{Q}_{t}^{s}$, and we let $\hat{Q}_{t}^*\equiv \hat{Q}_{t}^{t^0}$.

In this and the proceeding section, we abbreviate the following definitions from (\ref{eq:RemDdef}) for all $s,t$ with $s\ge t\ge t^0$: $\textnormal{D}_{s}^{t,*}\equiv \textnormal{D}_{s}^t(\hat{Q}_{s+1}^*,\hat{Q}_{s}^*)$ and $\Rem_{s}^{t,*}\equiv \Rem_{s}^t(\hat{Q}_{s}^*)$. Recall that $\textnormal{D}_{s}^{t,*}$ and $\Rem_{s}^{t,*}$ are functions mapping from the support of $\bar{H}_t\sim P$ to the real line, so that it makes sense to take $L^2(P)$ norms of these objects to quantify their magnitude.

\vspace{2pt}{\centering
[Algorithm~\ref{alg:sdr} about here.]\par
}\vspace{2pt}

We now analyze the procedure that targets the estimate of $Q_{s}$ towards $Q_t$, $s\ge t$. Define the data-dependent loss function
\begin{align*}
\mathscr{L}_{s}^t(\bar{h}_{s+1};\epsilon_{s}^t)= -&a_t\left\{\prod_{r=t+1}^{s} \frac{a_{r}}{\hat{\pi}_{r}}\right\}\left\{\hat{Q}_{s+1}^*\log \hat{Q}_{s}^{t+1,\epsilon_{s}^t}  + [1-\hat{Q}_{s+1}^*]\log\left[1-\hat{Q}_{s}^{t+1,\epsilon_{s}^t}\right]\right\},
\end{align*}
where above we suppressed the dependence of $\hat{\pi}_{r}$, $\hat{Q}_{s+1}^t$, and $\hat{Q}_{s}^{t+1,\epsilon_{s}^t}$ on $\bar{h}_{r}$, $\bar{h}_{s+1}$, and $\bar{h}_{s}$. One can show that $\E[\mathscr{L}_{s}^t(\bar{H}_{s+1};\epsilon_{s}^t)]$ is minimized at the $\bar{\epsilon}_{s}^t : \mathcal{H}_t\rightarrow\mathbb{R}$ satisfying
\begin{align}
\E\left[\left\{\prod_{r=t+1}^{s} \frac{A_{r}}{\hat{\pi}_{r}}\right\} \hat{Q}_{s}^{t+1,\bar{\epsilon}_{s}^t}(\bar{H}_{s})\middle|\bar{h}_t,A_t=1\right]=\E\left[\left\{\prod_{r=t+1}^{s} \frac{A_{r}}{\hat{\pi}_{r}}\right\} \hat{Q}_{s+1}^*\middle|\bar{h}_t,A_t=1\right], \label{eq:barepsdef}
\end{align}
where we note that a sufficient condition for this $\bar{\epsilon}_{s}^t$ to exist is that $P\{\hat{Q}_{s+1}^*(\bar{H}_{s+1})=0\}<1$ and $P\{\hat{Q}_{s+1}^*(\bar{H}_{s+1})=1\}<1$, i.e. that $\hat{Q}_{s+1}^*$ is not degenerate at zero or one. Define the conditional excess risk $\ER_{s}^t(\bar{h}_t)$ by $\E\left[\mathscr{L}_{s}^t(\bar{H}_{s+1};\hat{\epsilon}_{s}^t)-\mathscr{L}_{s}^t(\bar{H}_{s+1};\bar{\epsilon}_{s}^t)\middle|\bar{h}_t\right]$. The excess risk is defined as the average conditional excess risk, i.e. $\E[\ER_{s}^t(\bar{H}_t)]$. The upcoming lemma bounds the term $\norm{\textnormal{D}_{s}^{t,*}}$ from the upper bound in Lemma~\ref{lem:firstordexp} by the excess risk of the procedure for estimating $\hat{\epsilon}_{s}^t$ plus the deviation between the estimate of $Q_{s}$ targeted towards estimating $Q_{r}$, $r\ge t$, and the estimate of $Q_{s}$ that is targeted towards estimating $Q_{r}$, $r\ge t^0$. By the triangle inequality, controlling $\norm{\textnormal{D}_{s}^{t,*}}$ for all $s\ge t$ suffices to control $\norm{\sum_{s\ge t}\textnormal{D}_{s}^{t,*}}$ and, by Lemma~\ref{lem:firstordexp}, plays an important role in controlling $\norm{\hat{Q}_t^*-Q_t}$. 

\begin{theorem}[Upper bounding $\norm{\textnormal{D}_{s}^{t,*}}$ by an excess risk]\label{thm:excessriskbd}
Fix $s\ge t\ge t^0$. If $P\{\hat{Q}_{s+1}^*=0\}<1$, $P\{\hat{Q}_{s+1}^*=1\}<1$, and SP.$t^0$ holds, then, with probability one over draws $\bar{H}_t\sim P$, $\textnormal{D}_{s}^t(\hat{Q}_{s+1}^*,\hat{Q}_{s}^t)(\bar{H}_t)^2\lesssim \ER_{s}^t(\bar{H}_t)$. Furthermore, $\norm{\textnormal{D}_{s}^{t,*}}\lesssim \E\left[\ER_{s}^t(\bar{H}_t)\right]^{1/2} + \norm{\hat{Q}_{s}^* - \hat{Q}_{s}^t}$.
\end{theorem}
The above shows that $\norm{\textnormal{D}_{s}^{t,*}}$ converges to zero in probability if the excess risk of $\hat{\epsilon}_{s}^t$ for $\bar{\epsilon}_{s}^t$ converges to zero in probability and all targeting steps of the estimate of $Q_{s}$ that occur after the estimate is successfully targeted towards $Q_t$ (small excess risk) have little effect on the estimate. We will formally show that empirical risk minimizers satisfy this latter condition. Note also that, for $t=t^0$, i.e. for the final targeting step for the estimate of each $Q_{s}$, $\hat{Q}_{s}^*=Q_{s}^{t^0}$ by definition so that the latter term above is zero. Thus, $\norm{\textnormal{D}_{s}^{t^0,*}}^2$ converges to zero at least as quickly as does the excess risk $\E\left[\ER_{s}^{t^0}(\bar{H}_{t^0})\right]$. 

\section{Explicit guarantees for empirical risk minimizers}\label{sec:erm}
To establish concrete results about the iTMLE, it is easiest to analyze one particular class of estimators. Here we focus on estimators derived from empirical risk minimization (ERM) \citep[see, e.g.,][]{VandeGeer1990,Vapnik1991}.

We first give a brief review of ERM. Again suppose $Z\equiv(X,Y)\sim \nu$ and that $f_0$ minimizes the risk corresponding to some loss $\mathscr{L}$. An ERM attempts to estimate $f_0$ by letting $\hat{f}= \argmin_{f\in\mathcal{F}} \nu_n \mathscr{L}(\cdot;f)$, where $\nu_n \mathscr{L}(\cdot;f)$ is the empirical mean of $\mathscr{L}(Z;f)$ from an i.i.d. sample of size $n$ drawn from $\nu$, and $\mathcal{F}$ is some user-specified index set. While in practice this index set may depend on sample size, in this work we focus our analysis on a sample-size-independent $\mathcal{F}$. One could alternatively study a sieved estimator for which $\mathcal{F}$ grows with sample size.

For ease of notation, we assume that, when estimating each $Q_t$, the same class $\mathcal{F}^t$ is used to estimate both $\epsilon_{s}^t$, $s>t$, and also to estimate $Q_t$. We assume that $\mathcal{F}^t$ contains the trivial function mapping each $\bar{h}_t$ to zero. It is hard to imagine a useful class $\mathcal{F}^t$ that would violate this condition. We assume that each $\hat{\epsilon}_{s}^t$ is obtained via ERM, so that $\hat{\epsilon}_{s}^t\in \argmin_{\epsilon_{s}^t\in\mathcal{F}^t} P_n \mathscr{L}_{s}^t(\cdot;\epsilon_{s}^t)$. For each $t$, use the following correct specification assumptions.
\begin{enumerate}[leftmargin=*,widest={CS.$t$}]
	\item[\mylabel{it:cs}{CS.$t$)}] For $s\ge t$, the data-dependent functions $\bar{\epsilon}_{s}^t$ defined in (\ref{eq:barepsdef}) fall in $\mathcal{F}^t$ with probability approaching one.
\end{enumerate}
\begin{remark}[Alternative to CS]
If each $\hat{Q}_{s}^*$, $s>t$, has a (possibly misspecified) limit, then one can replace the condition that each $\bar{\epsilon}_{s}^t$, which relies on the sample-dependent estimate $\hat{Q}_{s}^*$, falls in $\mathcal{F}^t$ with the condition that the limit of $\bar{\epsilon}_{s}^t$, $s\ge t$, falls in $\mathcal{F}^t$. This is useful because, when $\hat{Q}_{s}^*$, $s>t$, is consistent, the limit of $\bar{\epsilon}_{s}^t$ is the constant function zero. Hence, for these $s$ our assumption that $\mathcal{F}^t$ contain this trivial function suffices. For $s=t$, replacing the above assumption by the assumption that the limit of $\bar{\epsilon}_{s}^t\in\mathcal{F}^t$ is like assuming OR.$t$ provided at least one of one of $\hat{Q}_{t+1}^*$ or $\hat{\pi}_{t+1}$ is consistent. 
\end{remark}
For ease of analysis, we also rely on the following assumption on each $\mathcal{F}^t$.
\begin{enumerate}[leftmargin=*,widest={BD}]
	\item[\mylabel{it:bd}{BD)}] For each $t\ge t^0$, the elements in $\mathcal{F}^t$ are uniformly bounded in $[-c,c]$, $c<\infty$.
\end{enumerate}
Each result also uses an empirical process condition to ensure that the class $\mathcal{F}^t$ is not too large, and also that the estimates of the propensity scores are well-behaved.
\begin{enumerate}[leftmargin=*,widest={DC}]
	\item[\mylabel{it:dc}{DC)}] For all $t\ge t^0$, $\mathcal{F}^{t}$ is a Donsker class \citep{vanderVaartWellner1996} and $\hat{\pi}_{t}$ belongs to a fixed Donsker class $\mathcal{D}_{t}$ with probability approaching one.
\end{enumerate}
\begin{remark}
Under the weaker condition that $\mathcal{F}^t$ is a Glivenko-Cantelli class and $\hat{\pi}_t$ belongs to a Glivenko-Cantelli class $\mathcal{D}_t$, one can obtain the same results as those that we will present in this section, with the only change being that each $o_P(n^{-1/4})$ is replaced by $o_P(1)$.
\end{remark}
\begin{theorem}[ERMs achieve SDR estimation of $Q_t$]\label{thm:erm}
If $t\ge t^0$ is such that CS.$t$, SP.$t$ holds with probability approaching one, BD, and DC, then, with probability approaching one,
\begin{align*}
\norm{\hat{Q}_t^*-Q_t}\lesssim\,& \sum_{s\ge t+1} \norm{\hat{Q}_{s}^*-Q_{s}}\norm{\hat{\pi}_{s}-\pi_{s}} + \sum_{s\ge t} \left(\norm{\hat{Q}_{s}^*-\hat{Q}_{s}^t} + \E[\mathcal{E}_{s}^t(\bar{H}_t)]^{1/2}\right).
\end{align*}
Furthermore, $\norm{\hat{Q}_t^*-Q_t}\lesssim \sum_{s\ge t+1} \norm{\hat{Q}_{s}^*-Q_{s}}\norm{\hat{\pi}_{s}-\pi_{s}} + o_P(n^{-1/4})$.
\end{theorem}
\begin{proof}[Proof of Theorem~\ref{thm:erm}]
We first use the bound in Lemma~\ref{lem:firstordexp}. Lemma~\ref{lem:targQkp} in Appendix~\ref{lem:targQkp} shows that each $\E[\mathcal{E}_{s}^t(\bar{H}_t)] = o_P(n^{-1/2})$. Cauchy-Schwarz and the fact that SP.$t$ holds with probability approaching one show that each $\norm{\Rem_{s}^t}$ is upper bounded by a constant times $\norm{\hat{Q}_{s}^*-Q_{s}}\norm{\hat{\pi}_{s}-\pi_{s}}$ with probability approaching one. The upcoming Theorem~\ref{thm:preserverate} shows that CS.$t$ and the other conditions of this theorem imply that targeting the estimates of each $Q_{s}$ has little effect on the estimates, i.e. $\norm{\hat{Q}_{s}^*-\hat{Q}_{s}^t}=o_P(n^{-1/4})$.
\end{proof}
We now make explicit the sense in which the above establishes the SDR property of our estimator. Suppose that $CS.t^0$ holds and, for each $t>t^0$, at least one of CS.$t$ and TM.$t$, i.e. $\norm{\hat{\pi}_t-\pi_t}=o_P(1)$, holds. A straightforward induction argument with inductive hypothesis ``CS.$t$ implies $\norm{\hat{Q}_t^*-Q_t}=o_P(1)$'' from $t=K,\ldots,t^0$ then shows that $\norm{\hat{Q}_{t^0}^*-Q_{t^0}}=o_P(1)$. Thus, our approach is SDR once we replace OR.$t$ by the related, but somewhat more technical, condition CS.$t$. If each $\norm{\hat{\pi}_t-\pi_t}$ is $o_P(n^{-1/4})$, which is achievable if $\hat{\pi}_t$ is an ERM from a correctly specified Donsker class, then the same induction argument holds, but now with the $o_P(1)$ replaced by $o_P(n^{-1/4})$ so that $\norm{\hat{Q}_{t^0}^*-Q_{t^0}}=o_P(n^{-1/4})$.

We now show that, for $s\ge t\ge t^0$, targeting the estimate  of $Q_{s}$ towards all $r=t-1,\ldots,t^0$ has little effect on the estimate if CS.$t$ holds.
\begin{theorem}\label{thm:preserverate}
Fix a $t\ge t^0$ for which CS.$t$ holds and let $s\ge t$. If SP.$t^0$ holds with probability approaching one, BD, and DC, then $\norm{\hat{Q}_{s}^*-\hat{Q}_{s}^t} = o_P(n^{-1/4})$.
\end{theorem}

\begin{remark}[Improved rate quantification]
Under entropy integral bounds on the Donsker classes in DC, local properties of the empirical process would yield a better understanding of the $o_P(n^{-1/4})$ rates above. One could use local maximal inequalities for bracketing entropy \citep{vanderVaartWellner1996} or uniform entropy \citep{vanderVaart&Wellner2011}. 
\end{remark}

\begin{remark}[Conjectured rate optimality when DR terms are small]\label{rem:rightrate}
Suppose that CS.$t$ and TM.$t$ hold at every time point $t>t^0$ and that CS.$t^0$ holds. Further suppose that $\norm{\hat{\pi}_t-\pi_t}=o_P(n^{-1/4})$ for each $t>t^0$, which will be the case if each $\pi_t$ is estimated using an ERM over a correctly specified Donsker class. In this case a simple induction argument shows that
\begin{align*}
\norm{\hat{Q}_{t^0}^*-Q_{t^0}}&\lesssim 
\sum_{t\ge t^0} \left\{(P_n-P)\left[\mathscr{L}_t^{t^0}(\cdot\,;\bar{\epsilon}_t^{t^0})-\mathscr{L}_t^{t^0}(\cdot\,;\hat{\epsilon}_t^{t^0})\right]\right\}^{1/2} + o_P(n^{-1/2}).
\end{align*}
The first inequality holds whether or not CS.$t^0$ is true, but the second uses CS.$t^0$. In particular, the second inequality holds because each $\hat{\epsilon}_t^{t^0}$ is an ERM over $\mathcal{F}^{t^0}$ and each $\bar{\epsilon}_t^{t^0}\in \mathcal{F}^{t^0}$ by CS.$t^0$. Even in a correctly specified parametric model, the leading term is only $O_P(n^{-1/2})$, so the leading sum above is always expected to dominate. As the rate of convergence of the empirical processes can be controlled by the size of the class $\mathcal{F}^{t^0}$, we see that $\norm{\hat{Q}_{t^0}^*-Q_{t^0}}$ converges to zero at the rate dictated by size of the class $\mathcal{F}^{t^0}$, where the size of the class can be quantified using metric entropy \citep{vanderVaartWellner1996}.

Compare this to earlier sequential regression procedures, whose rate of convergence is typically dominated by the size of the largest $\mathcal{F}^t$, $t\ge t^0$. As $\bar{H}_{t^0}$ is necessarily lower dimensional than $\bar{H}_t$, $t>t^0$, we would typically expect that $\mathcal{F}^{t^0}$ has a smaller entropy integral than $\mathcal{F}^t$. Hence, we expect that traditional sequential regression procedures have rate dominated by the size of $\mathcal{F}^K$. It seems likely that this fact enables the construction of confidence sets for $Q_{t^0}$. We will examine this further in future works.
\end{remark}

\begin{remark}[Asymptotic linearity]
A stronger result than that of Theorem~\ref{thm:erm} can be obtained if $t^0=0$ so that $Q_{t^0}$ is a real number. Suppose that the conditions of Theorem~\ref{thm:erm} hold and there exists a function $\IF : \bar{\mathcal{H}}_{K+1}\rightarrow\mathbb{R}$ such that $\norm{\IF-\sum_{t=0}^K \widehat{\IF}_t^{\hat{\epsilon}_t^{0}}}=o_P(1)$, where
\begin{align*}
\widehat{\IF}_t^{\epsilon_t^{0}}(\bar{h}_{t+1})\equiv \left\{\prod_{s=1}^{t}\frac{a_{s}}{\hat{\pi}_{s}(\bar{h}_{s})}\right\}\left\{\hat{Q}_{t+1}^*(\bar{h}_{t+1})-\hat{Q}_{t}^{1,\epsilon_{t}^{0}}(\bar{h}_{t})\right\}.
\end{align*}
Here, $\hat{Q}_{t}^{1,\epsilon_{t}^{0}}= \hat{Q}_{t}^{t^0+1,\epsilon_{t}^{t^0}}$ is defined in Algorithm~\ref{alg:sdr}. In particular, the fact that, for each $t$, $\{\mathscr{L}_{t}^{0}(\cdot,\epsilon_{t}^{0}) : \epsilon_{t}^{0}\in\mathcal{F}^{0}=\mathbb{R}\}$ is a parametric class ensures that $0=\left.\frac{\partial}{\partial\epsilon_{t}^{0}} P_n \mathscr{L}_{t}^{0}(\cdot,\epsilon_{t}^{0})\right|_{\epsilon_{t}^{0}=\hat{\epsilon}_{t}^{0}} = P_n \widehat{\IF}_{t}^{\hat{\epsilon}_{t}^{0}}$. Noting that $\textnormal{D}_{t}^{0}=-P\widehat{\IF}_{t}^{\hat{\epsilon}_{t}^{0}}$,
\begin{align*}
\sum_{t=0}^K \textnormal{D}_{t}^{0}=(P_n-P)\sum_{t=0}^K\widehat{\IF}_{t}^{\hat{\epsilon}_{t}^{0}} = (P_n-P)\IF - (P_n-P)\left[\IF - \sum_{t=0}^K\widehat{\IF}_{t}^{\hat{\epsilon}_{t}^{0}}\right].
\end{align*}
By $\norm{\IF-\sum_{t=0}^{K} \widehat{\IF}_{t}^{\hat{\epsilon}_{t}^{0}}}=o_P(1)$, DC, permanence properties of Donsker classes \citep{vanderVaartWellner1996}, \ref{it:bd}, and \ref{it:sp}, the latter term is $o_P(n^{-1/2})$. Thus, by Lemma~\ref{lem:firstordexp},
\begin{align*}
\hat{Q}_{0}^*-Q_{0}&= (P_n-P)\IF + \sum_{t=1}^{K} \Rem_{t}^{0} + o_P(n^{-1/2}),
\end{align*}
and so $\hat{Q}_{0}^*$ is an asymptotically linear estimator of $Q_{0}$ with influence function $\IF$ if the remainder is $o_P(n^{-1/2})$. If one has not used known values of $\pi_{t}$ or correctly specified a parametric model for each $\pi_{t}$, $t\ge 1$, then often $\IF$ is the canonical gradient in the nonparametric model \citep{Pfanzagl1990}.
\end{remark}

\section{Simulation studies}\label{sec:sim}
\noindent\textit{Simulation setup.} We conduct a simulation study that evaluates the finite sample behavior
of the two SDR methods presented. All simulations report (i) the
mean-squared error (MSE) for the outcome regression that conditions on
baseline covariates only, i.e. $Q_{1}$; (ii) the bias and the coverage
of a two-sided 95\% confidence interval given by the various estimators
for the marginal parameter $Q_{0}$. As most existing methods are designed
to focus on estimating $Q_{0}$, they will not necessarily perform well
for $Q_{1}$.

We compare the performance of the following estimators: LTMLE, doubly robust unbiased transformation
(DR Transform), and iTMLE. We also implement a na\"{i}ve plug-in estimator (Direct Plugin),
that lets $\hat{Q}_{K+1}(\bar{H}_{K+1})\equiv L_{K+1}$, and recursively
from $t=K,\ldots,0$, regresses $\hat{Q}_{t+1,i}$ against $\bar{H}_{t,i}$
among all individuals $i$ with $A_{t,i}=1$. The na\"{i}ve plug-in estimator
yields estimates of $Q_{0}$ and $Q_{1}$, though does not yield a confidence
interval for $Q_{0}$. For estimation of the marginal mean parameter
$Q_{0}$, we also evaluate the bias for the inverse probability weighted
estimator (IPW). The performance of these estimators is evaluated under
various model specification scenarios for the outcome regressions $\hat{Q}_{t}$
and the propensity scores $\hat{\pi}_{t}$, as described below. We do not evaluate the performance of the Bang \& Robins (BR) estimator
in its original form because LTMLE can be viewed as its robust extension.
Some of the data generating distributions used in this simulation
study would yield an inconsistent BR, thus not providing a fair comparison. The performance of the estimators is evaluated
over $1000$ Monte Carlo draws. All simulations
are carried out in \texttt{R} \citep{R2016} using the packages
\emph{$\mathtt{simcausal}$} \citep{Sofryginetal2017} and \emph{$\mathtt{stremr}$}
\citep{Sofryginetal2016}. The code for our simulation
is available in a github repository (\url{http://github.com/osofr/SDRsimstudy}).

Our simulation study consists of two scenarios. The data-generating distributions for Simulations 1 and 2 are described in detail
in Appendix~\ref{app:simDGD}. Here we give an overview of the simulation methods. Simulation 1
is a proof of concept with a simple longitudinal structure
and 3 time-point interventions, i.e., $K=3$. For this scenario,
the estimates are evaluated from a sample of $n=500$ i.i.d. units. The
outcome regressions $Q_{t}$ and propensity scores $\pi_{t}$ are estimated
using the main terms logistic regressions. The estimator of the outcome regression $Q_{1}$ is always correctly
specified in all four 
scenarios.  Each estimator is evaluated
based on the following four regression specification scenarios of the remaining outcome regressions and propensity scores: \emph{Qc.gc},
when all $Q_{t}$ and $\pi_{t}$ are based on correctly specified regressions;
\emph{Qi.gc}, all $Q_{t}$ for
$t>1$ are incorrect, and $\pi_{t}$ are correctly specified for all $t$;
\emph{Qc.gi}, when $Q_{t}$ are correctly specified for all $t$, while
$\pi_{t}$ are incorrect for all $t$;\emph{ Qi.gi}, $Q_{t}$ are incorrectly specified for $t>1$ and $g_{t}$
is incorrectly specified for all $t$.

The second simulation scenario (Simulation 2) is based on a 5 time-point longitudinal structure, i.e, $K=5$, with $n=5000$. 
The types of regressions considered in this simulation include the four
scenarios from Simulation 1 (\emph{Qc.gc,} \emph{Qi.gc}, \emph{Qc.gi} and
\emph{Qi.gi)}, except that the estimation of $Q_{t}$ is based on non-parametric
regression approaches (details below). We define the ``correct''
estimation scenario for $Q_{t}$ (i.e, \emph{ Qc.gc} and \emph{Qc.gi)} by including all the relevant time-varying and baseline covariates,
whereas the ``incorrect'' estimation scenario
for $Q_{t}$ means that we exclude key time-varying
or baseline covariates. We estimate each $\pi_{t}$
via the main-terms logistic regression. The incorrect estimator $\hat{\pi}_{t}$
of $\pi_{t}$ is obtained by running an intercept-only logistic regression.
We also consider an additional scenario (\emph{QSDR.gSDR}), where
$Q_{5}$ is incorrect, while $Q_{t}$ are correct for all $t<5$, and,
conversely, $\pi_{5}$ is correct, while $\pi_{t}$ are incorrect for all
$t<5$. This scenario mimics data for which the last outcome regression is a
high-dimensional and biologically complex mechanism and is unlikely to be correctly specified, while the exposure mechanism at the last time-point is known.

For Simulation 2, the non-parametric estimation of $Q_{t}$ is based on
a discrete super-learner \citep{vanderLaan&Polley&Hubbard07}.
The ensemble library of candidate learners includes
18 estimators from $\mathtt{xgboost}$ \texttt{R} package \citep{Chen&Guestrin2016}, 
as well as a main-terms logistic
regression (GLM). The best performing model in the ensemble is selected
via 5-fold cross-validation. We found that using the ensemble of highly
data-adaptive $\mathtt{xgboost}$ learners for all $Q_{t}$ 
was prone to overfitting. To mitigate this overfitting we
employ $\mathtt{xgboost}$-based learners only for estimating $Q_{5}$
and $Q_{4}$, and use the GLM for estimating $Q_{3},Q_{2}$
and $Q_{1}$.

In both simulation scenarios, the iTMLE targeting steps 
are
based on super-learner ensembles that include
3 GBMs from $\mathtt{xgboost}$ \texttt{R} package, a main
terms logistic regression, a univariate intercept-only logistic regression
and an empty learner that does not update $\hat{Q}_{t}$. The
targeted iTMLE update is then defined by the convex combination of predictions
from each learner in the super-learner ensemble, where this combination is fitted using the novel cross-validation
scheme presented in Appendix~\ref{app:cv}. 

The regression specification for DR Transform relies on the same
estimation approaches as for $Q_{t}$ in Simulation 1 and 2.
However, since the transformed estimates $\hat{\Gamma}_{t,i}$ often result
in some values being outside of $(0,1)$, the standard statistical \texttt{R}
software, such as GLM, produces an error. To overcome this, we
modified the \texttt{R} package $\mathtt{xgboost}$ to produce valid regression
estimates with DR transformed outcomes, even if
they fall outside of $(0,1)$.

\vspace{6pt}\noindent\textit{Simulation results.} The simulation results for the relative MSE estimation of $Q_{1}$ for
Simulation 1 and 2 are presented in Figure~\ref{fig:simres.all.MSE}. These
results clearly demonstrate that, depending on the scenario, the iTMLE and DR Transform either outperform
or perform comparably to Direct Plugin and LTMLE. This figure also shows that the iTMLE outperforms DR transform in terms of MSE in Simulation 1 for \emph{Qi.gc} and \emph{Qi.gi}, i.e. the two settings where the later time point outcome regressions are misspecified. The Simulation
1 and 2 results for the relative absolute bias in estimation of $Q_{0}$ are presented in Figure~\ref{fig:simres.all.BIAS}. While overall
the performance the LTMLE, iTMLE, and DR Transform is similar across different
scenarios, the notable exceptions are the scenario \emph{Qi.gc} in Simulation
1, where DR Transform appears to outperform other methods, the scenario
\emph{Qc.gi} in Simulation 2, where DR Transform outperforms the rest,
and the scenario \emph{QSDR.gSDR} in Simulation 2, where both SDR methods
outperform the LTMLE. The simulation results for the coverage and mean length of the two-sided 95\% CIs for $Q_{0}$ in Simulation 1 and 2 are presented in Figure~\ref{fig:simres.all.coverCIlen}.
The confidence interval coverage and width appear to be comparable between
the two SDR methods and the LTMLE. The only exception is for the \emph{QSDR.gSDR}
scenario, where the LTMLE has roughly 10\% coverage while the SDR approaches 
achieve nearly the nominal coverage level similar mean confidence interval widths.

\vspace{2pt}{\centering
[Figures~\ref{fig:simres.all.MSE}, \ref{fig:simres.all.BIAS}, and \ref{fig:simres.all.coverCIlen} about here.]\par
}\vspace{2pt}

\section{Discussion}\label{sec:disc}
We have discussed the sequentially doubly robust estimation of the longitudinal G-formula. This form of robustness allows inconsistent estimation of either the treatment/censoring mechanism or outcome regression at each time point. We presented a general SDR estimation strategy, referred to as the iTMLE. This procedure is iterative, leveraging the SDR property from temporally subsequent outcome regressions to ensure the SDR property at the current outcome regression. We presented a high-level argument supporting the SDR nature of a general iTMLE, and formally established that our estimation scheme is SDR when implemented with empirical risk minimization. In practice, we believe the ERM procedure is prone to overfitting, and suggest using the cross-validation selector presented in the Supplement. Beyond the added robustness of our new estimator, we argued why the iTMLE more appropriately accounts for the dimension of the outcome regression problem than typical sequential regression procedures, thereby leading to better finite sample behavior.

To improve the readability of this manuscript, we focused on outcomes bounded in $[0,1]$, and, as a consequence, all regressions were based upon a cross-entropy loss function. Like targeted minimum-loss based estimation, this method immediately extends beyond binary outcomes by choosing a different loss function. To simplify the presentation of our technical analysis of the ERM special case, in Section~\ref{sec:erm}, we also assumed that the outcome regressions were bounded away from zero and one. While time-to-event outcomes, discussed in Appendix~\ref{app:rtcens}, may seem to be excluded by this assumption, this is only an artificial restriction. The theory presented can be shown to remain valid even without this condition.

We expect our method to enable the construction of confidence sets and bands for time $t$ outcome regressions that shrink at the rate dictated by the entropy of the class used to estimate outcome regression $t$ rather than the entropy of the largest class used to estimate outcome regressions $s\ge t$. We will explore this in future work.

\vspace{6pt}\noindent\textbf{Acknowledgements.}
This work was partially supported by the National Institute of Allergy and Infectious Disease at the National Institutes of Health under award UM1 AI068635.

\bibliographystyle{Chicago}
\bibliography{persrule}

\section{Proofs}\label{app:proofs}
\begin{proof}[Proof of Lemma~\ref{lem:firstordexp}]
Note that
\begin{align*}
\hat{Q}_t^t(\bar{h}_t)-Q_t(\bar{h}_t)= -\E\left[\frac{A_t}{\pi_t}\left\{\hat{Q}_{t+1}-\hat{Q}_t\right\}\middle|\bar{h}_t\right] + \E\left[\hat{Q}_{t+1} - Q_{t+1}\middle|\bar{h}_t,A_t=1\right].
\end{align*}
The leading term is $\textnormal{D}_t^t(\bar{h}_t)$. To see that the latter term equals $\sum_{s=t+1}^K [\textnormal{D}_{s}^t(\bar{h}_t) + \Rem_{s}(\bar{h}_t)]$, recursively (from $s=t$ to $s=K-1$) apply the following relationship to the inner expectation in the final term of
\begin{align*}
\E\left[\hat{Q}_{s+1} - Q_{s+1}\middle|\bar{h}_{s},a_{s}\right]
=\,& -\E\left[\frac{A_{s+1}}{\hat{\pi}_{s+1}}\left\{\hat{Q}_{s+2}-\hat{Q}_{s+1}\right\}\middle|\bar{h}_{s},a_{s}\right] \\
&+ \E\left[\left\{1-\frac{\pi_{s+1}}{\hat{\pi}_{s+1}}\right\}\left\{\hat{Q}_{s+1} - Q_{s+1}\right\}\middle|\bar{h}_{s},a_{s}\right] \\
&+\E\left[\frac{A_{s+1}}{\hat{\pi}_{s+1}}\E\left[\hat{Q}_{s+2} - Q_{s+2}\middle|\bar{H}_{s+1},A_{s+1}\right]\middle|\bar{h}_{s},a_{s}\right],
\end{align*}
where the recursion ends at $s=K-1$ because $\hat{Q}_{K+1}=Q_{K+1}$.
\end{proof}

\begin{proof}[Proof of Theorem~\ref{thm:excessriskbd}]
Fix $\bar{h}_t$. Define
\begin{align*}
G_{\bar{h}_t}(\varepsilon)&\equiv \E\left[\mathscr{L}_{s}^t(\bar{H}_{s+1};\bar{h}_t'\mapsto \varepsilon)\middle|\bar{h}_t\right].
\end{align*}
Let $\dot{G}_{\bar{h}_t}(\varepsilon)\equiv \frac{\partial}{\partial \varepsilon} G_{\bar{h}_t}(\varepsilon)$. The chain rule shows that $\frac{\partial [\dot{G}_{\bar{h}_t}(\varepsilon)]^2}{\partial G_{\bar{h}_t}(\varepsilon)} = 2\ddot{G}_{\bar{h}_t}(\varepsilon)$, where $\ddot{G}_{\bar{h}_t}(\varepsilon)\equiv \frac{\partial}{\partial\varepsilon} \dot{G}_{\bar{h}_t}(\varepsilon)$. By the mean value theorem, there exists a $c$ in the range of $2\ddot{G}_{\bar{h}_t}(\cdot)$ such that
\begin{align*}
\dot{G}_{\bar{h}_t}[\hat{\epsilon}_{s}^t(\bar{h}_t)]^2 - \dot{G}_{\bar{h}_t}[\bar{\epsilon}_{s}^t(\bar{h}_t)]^2&= c\left\{G_{\bar{h}_t}[\hat{\epsilon}_{s}^t(\bar{h}_t)]-G_{\bar{h}_t}[\bar{\epsilon}_{s}^t(\bar{h}_t)]\right\}
\end{align*}
As $\bar{\epsilon}_{s}^t$ is the risk minimizer over all functions $\bar{\epsilon}_{s}^t : \bar{\mathcal{H}}_t\rightarrow\mathbb{R}$, $\dot{G}_{\bar{h}_t}[\bar{\epsilon}_{s}^t(\bar{h}_t)] = 0$. Straightforward calculations show that
\begin{align*}
2\ddot{G}_{\bar{h}_t}(\varepsilon) = 2a_t\left\{\prod_{r=t+1}^{s} \frac{A_{r}}{\hat{\pi}_{r}}\right\} \hat{Q}_{s}^{t+1,\bar{h}_t'\mapsto \varepsilon}(1-\hat{Q}_{s}^{t+1,\bar{h}_t'\mapsto \varepsilon}),
\end{align*}
which is uniformly upper bounded by $\delta^{s-t}/2$. Furthermore, because $\bar{\epsilon}_{s}^t$ is a risk minimizer, $G_{\bar{h}_t}[\hat{\epsilon}_{s}^t(\bar{h}_t)]-G_{\bar{h}_t}[\bar{\epsilon}_{s}^t(\bar{h}_t)]\ge 0$. Further note that $\dot{G}_{\bar{h}_t}[\hat{\epsilon}_{s}^t(\bar{h}_t)]^2 = \pi_t(\bar{h}_t)^2 \textnormal{D}_{s}^t(\hat{Q}_{s+1}^*,\hat{Q}_{s}^t)(\bar{h}_t)^2$, and also that $\textnormal{D}_{s}^t(\hat{Q}_{s+1}^*,\hat{Q}_{s}^t)(\bar{h}_t)^2\lesssim \textnormal{D}_{s}^t(\hat{Q}_{s+1}^*,\hat{Q}_{s}^t)(\bar{h}_t)^2$. Hence, we have shown that
\begin{align*}
\textnormal{D}_{s}^t(\hat{Q}_{s+1}^*,\hat{Q}_{s}^t)(\bar{h}_t)^2&\lesssim \left\{G_{\bar{h}_t}[\hat{\epsilon}_{s}^t(\bar{h}_t)]-G_{\bar{h}_t}[\bar{\epsilon}_{s}^t(\bar{h}_t)]\right\} = \ER_{s}^t(\bar{h}_t).
\end{align*}
This yield the almost sure pointwise bound. Take an expectation over $\bar{H}_t\sim P$ on both sides, taking the square root, and applying the triangle inequality shows that
\begin{align*}
\norm{\textnormal{D}_{s}^{t,*}}&\le \norm{\textnormal{D}_{s}^t(\hat{Q}_{s+1}^*,\hat{Q}_{s}^t)} + \norm{\textnormal{D}_{s}^{t,*}-\textnormal{D}_{s}^t(\hat{Q}_{s+1}^*,\hat{Q}_{s}^t)} \\
&\lesssim \E\left[\ER_{s}^t(\bar{h}_t)\right]^{1/2} + \norm{\textnormal{D}_{s}^{t,*}-\textnormal{D}_{s}^t(\hat{Q}_{s+1}^*,\hat{Q}_{s}^t)}.
\end{align*}
The positivity assumption shows that the latter term upper bounds by a positive constant (relying on $\delta$ only) times $\norm{\hat{Q}_{s}^*-\hat{Q}_{s}^t}$. 
\end{proof}

\begin{lemma}[Targeting $\hat{Q}_{s}^{t+1}$ makes the excess risk small]\label{lem:targQkp}
Fix $t$. Under the conditions of Theorem~\ref{thm:erm},
\begin{align*}
\E[\mathcal{E}_{s}^t(\bar{H}_t)] = o_P(n^{-1/2})\textnormal{ for each $s\ge t$.}
\end{align*}
\end{lemma}
\begin{proof}[Proof of Lemma~\ref{lem:targQkp}]
We use empirical process notation so that, for a distribution $\nu$ and function $f$, $\nu f=\E_{\nu}[f(Z)]$. As $\hat{\epsilon}_{s}^t$ is an empirical risk minimizer, \ref{it:cs} implies that $P_n \mathscr{L}_{s}^t(\cdot\,;\hat{\epsilon}_{s}^t)\le P_n\mathscr{L}_{s}^t(\cdot\,;\bar{\epsilon}_{s}^t)$. Hence,
\begin{align*}
\E[\mathcal{E}_{s}^t(\bar{H}_t)]&= P\left[\mathscr{L}_{s}^t(\cdot;\hat{\epsilon}_{s}^t)-\mathscr{L}_{s}^t(\cdot;\bar{\epsilon}_{s}^t)\right]\le (P_n-P)\left[\mathscr{L}_{s}^t(\cdot\,;\bar{\epsilon}_{s}^t)-\mathscr{L}_{s}^t(\cdot\,;\hat{\epsilon}_{s}^t)\right].
\end{align*}
The remainder of the proof shows that the right-hand side is $o_P(n^{-1/2})$. By \ref{it:bd}, \ref{it:cs}, and permanence properties of Donsker classes \citep[e.g., Chapter 2.10 in][]{vanderVaartWellner1996}, the right-hand side is $O_P(n^{-1/2})$. Using the bounds on $\hat{Q}_{s}^{t+1}$, $s>t$, \ref{it:bd}, and \ref{it:sp}, standard arguments used to show that the cross-entropy loss is quadratic (see, e.g., Lemma~2 in \citealp{vanderLaan&Dudoit&Keles04}) show that
\begin{align}
P\left[\left\{\mathscr{L}_{s}^t(\cdot;\hat{\epsilon}_{s}^t)-\mathscr{L}_{s}^t(\cdot;\bar{\epsilon}_{s}^t)\right\}^2\right]\lesssim P\left[\mathscr{L}_{s}^t(\cdot;\hat{\epsilon}_{s}^t)-\mathscr{L}_{s}^t(\cdot;\bar{\epsilon}_{s}^t)\right]. \label{eq:qlquad}
\end{align}
We have already shown that the left-hand side is $O_P(n^{-1/2})$. Combining this with DC, permanence properties of Donsker classes, and the asymptotic equicontinuity of Donsker classes \citep[e.g., Lemma 19.24 in][]{vanderVaart98}, $(P_n-P)\left[\mathscr{L}_{s}^t(\cdot\,;\bar{\epsilon}_{s}^t)-\mathscr{L}_{s}^t(\cdot\,;\hat{\epsilon}_{s}^t)\right]=o_P(n^{-1/2})$.
\end{proof}

\begin{proof}[Proof of Theorem~\ref{thm:preserverate}]
Let $t,s$ satisfy the conditions of the theorem. We give proof by induction on $r=t,t-1,\ldots,t^0$. The inductive hypothesis at $r=t^0$ includes our desired result that $\norm{\hat{Q}_{s}^*-\hat{Q}_{s}^t} = o_P(n^{-1/4})$.\\
\textbf{Induction Hypothesis: IH($r$).}  $\norm{\hat{Q}_{s}^t-\hat{Q}_{s}^{r}}=o_P(n^{-1/4})$ and $\E[\mathcal{E}_{s}^{r}(\bar{H}_{r})]=o_P(n^{-1/2})$.\\
\textbf{Base Case:} $r=t$. $\norm{\hat{Q}_{s}^t-\hat{Q}_{s}^{r}}=0$, so is $o_P(n^{-1/4})$ with much to spare. Lemma~\ref{lem:targQkp} shows that CS.$t$ plus the other conditions of this theorem imply that $\E[\mathcal{E}_{s}^{r}(\bar{H}_{r})]=o_P(n^{-1/2})$.\\
\textbf{Induction Step:} Suppose IH($r+1$) holds. By the triangle inequality,
\begin{align*}
\norm{\hat{Q}_{s}^t-\hat{Q}_{s}^{r}}&\le \norm{\hat{Q}_{s}^t-\hat{Q}_{s}^{r+1}} + \norm{\hat{Q}_{s}^{r+1}-\hat{Q}_{s}^{r}}.
\end{align*}
By IH($r+1$), the leading term above is $o_P(n^{-1/4})$. In the remainder we establish that $\norm{\hat{Q}_{s}^{r+1}-\hat{Q}_{s}^{r}}=o_P(n^{-1/4})$, and along the way we also establish that $\E[\mathcal{E}_{s}^{r}(\bar{H}_{r})]=o_P(n^{-1/2})$. By the Lipschitz property of the expit function, $\norm{\hat{Q}_{s}^{r+1}-\hat{Q}_{s}^{r}}\lesssim \norm{\hat{\epsilon}_{s}^{r}}$, and thus it suffices to bound $\norm{\hat{\epsilon}_{s}^{r}}$ to establish the first part of IH($r$).

We start by giving a useful upper bound for $\norm{\epsilon-\bar{\epsilon}_{s}^{r}}$ for a general function $\epsilon : \bar{\mathcal{H}}_{r}\rightarrow\mathbb{R}$ that falls in $L^2(P)$. Because $\bar{\epsilon}_{s}^{r}$ is a risk minimizer over all functions mapping from $\bar{\mathcal{H}}_{r}$ to $\mathbb{R}$, a $\bar{H}_{r}$-pointwise second-order Taylor expansion shows that, for some $\tilde{\epsilon} : \bar{\mathcal{H}}_{r}\rightarrow\mathbb{R}$ that falls in between $\epsilon$ and $\bar{\epsilon}_{s}^{r}$,
\begin{align*}
\E&\left[\mathscr{L}_{s}^{r}(\bar{H}_{s+1};\epsilon)-\mathscr{L}_{s}^{r}(\bar{H}_{s+1};\bar{\epsilon}_{s}^{r})\right] \\
&= \E\left[\E\left[\mathscr{L}_{s}^{r}(\bar{H}_{s+1};\epsilon)-\mathscr{L}_{s}^{r}(\bar{H}_{s+1};\bar{\epsilon}_{s}^{r})\middle|\bar{H}_{r}\right]\right] \\
&= \frac{1}{2}\E\left[\{\epsilon(\bar{H}_{r})-\bar{\epsilon}_{s}^{r}(\bar{H}_{r})\}^2 \E\left[A_{r}\left(\prod_{r'=r+1}^{s}\frac{A_{r'}}{\hat{\pi}_{r'}}\right)\hat{Q}_{s}^{r+1,\tilde{\epsilon}}(1-\hat{Q}_{s}^{r+1,\tilde{\epsilon}})\middle|\bar{H}_{r}\right]\right] \\
&\ge c\norm{\epsilon-\bar{\epsilon}_{s}^{r}}^2
\end{align*}
for an appropriately specified constant $c>0$, where we used BD. The triangle inequality and two applications of the preceding display (at $\epsilon=0$ and $\epsilon=\hat{\epsilon}_{s}^{r}$) show that
\begin{align*}
\norm{\hat{\epsilon}_{s}^{r}}\le \norm{\bar{\epsilon}_{s}^{r}} + \norm{\hat{\epsilon}_{s}^{r}-\bar{\epsilon}_{s}^{r}}\lesssim \E\left[\mathscr{L}_{s}^{r}(\bar{H}_{s+1};0)-\mathscr{L}_{s}^{r}(\bar{H}_{s+1};\bar{\epsilon}_{s}^{r})\right]^{1/2} + \E[\mathcal{E}_{s}^{r}(\bar{H}_{r})]^{1/2}.
\end{align*}
The square of the latter term upper bounds as follows:
\begin{align}
\E[\mathcal{E}_{s}^{r}(\bar{H}_{r})]&= P\left[\mathscr{L}_{s}^{r}(\cdot;\hat{\epsilon}_{s}^{r})-\mathscr{L}_{s}^{r}(\cdot;0)\right] + P\left[\mathscr{L}_{s}^{r}(\cdot;0)-\mathscr{L}_{s}^{r}(\cdot;\bar{\epsilon}_{s}^{r})\right] \nonumber \\
&\le P\left[\mathscr{L}_{s}^{r}(\cdot;\hat{\epsilon}_{s}^{r})-\mathscr{L}_{s}^{r}(\cdot;0)\right] - (P_n-P)\left[\mathscr{L}_{s}^{r}(\cdot;0)-\mathscr{L}_{s}^{r}(\cdot;\bar{\epsilon}_{s}^{r})\right] \nonumber \\
&= P\left[\mathscr{L}_{s}^{r}(\cdot;\hat{\epsilon}_{s}^{r})-\mathscr{L}_{s}^{r}(\cdot;0)\right] + \left|(P_n-P)\left[\mathscr{L}_{s}^{r}(\cdot;0)-\mathscr{L}_{s}^{r}(\cdot;\bar{\epsilon}_{s}^{r})\right]\right|, \label{eq:excessriskinduction}
\end{align}
where the inequality uses that the constant function zero is in $\mathcal{F}^{r}$ and the latter equality uses that $\bar{\epsilon}_{s}^{r}$ is the true risk minimizer and $\hat{\epsilon}_{s}^{r}$ is the empirical risk minimizer. The subadditivity of $x\mapsto x^{1/2}$ thus yields that
\begin{align*}
\norm{\hat{\epsilon}_{s}^{r}}\lesssim\,& \E\left[\mathscr{L}_{s}^{r}(\bar{H}_{s+1};0)-\mathscr{L}_{s}^{r}(\bar{H}_{s+1};\bar{\epsilon}_{s}^{r})\right]^{1/2} + \left|(P_n-P)\left[\mathscr{L}_{s}^{r}(\cdot;0)-\mathscr{L}_{s}^{r}(\cdot;\bar{\epsilon}_{s}^{r})\right]\right|^{1/2}.
\end{align*}
The square of the first term above upper bounds as follows:
\begin{align}
\E&\left[\mathscr{L}_{s}^{r}(\bar{H}_{s+1};0)-\mathscr{L}_{s}^{r}(\bar{H}_{s+1};\bar{\epsilon}_{s}^{r})\right] \nonumber \\
&= \E\left[\frac{A_{r}}{\hat{\pi}_{r+1}}\E\left[\mathscr{L}_{s}^{r+1}(\bar{H}_{s+1};\hat{\epsilon}_{s}^{r+1})-\mathscr{L}_{s}^{r+1}(\bar{H}_{s+1};\hat{\epsilon}_{s}^{r+1} + \bar{\epsilon}_{s}^{r})\middle|\bar{H}_{r+1}\right]\right] \nonumber \\
&\lesssim \E\left[\frac{A_{r}}{\hat{\pi}_{r+1}}\E\left[\mathscr{L}_{s}^{r+1}(\bar{H}_{s+1};\hat{\epsilon}_{s}^{r+1})-\mathscr{L}_{s}^{r+1}(\bar{H}_{s+1};\bar{\epsilon}_{s}^{r+1})\middle|\bar{H}_{r+1}\right]\right] \nonumber \\
&\lesssim \E[\mathcal{E}_{s}^{r+1}(\bar{H}_{r})]. \label{eq:riskub}
\end{align}
The equality is an algebraic identity, the first inequality uses that $\bar{\epsilon}_{r}^{t+1}$ is the risk minimizer among all functions mapping from $\bar{\mathcal{H}}_{r+1}\rightarrow\mathbb{R}$, and the second inequality uses SP.$t^0$. By IH($r+1$), the right-hand side is $o_P(n^{-1/2})$. Returning to the preceding display,
\begin{align*}
\norm{\hat{\epsilon}_{s}^{r}}\lesssim\,& \left|(P_n-P)\left[\mathscr{L}_{s}^{r}(\cdot;0)-\mathscr{L}_{s}^{r}(\cdot;\bar{\epsilon}_{s}^{r})\right]\right|^{1/2} + o_P(n^{-1/4}).
\end{align*}
By DC, BD, SP.$t^0$, and Lemma~19.24 of \cite{vanderVaart98}, the former term on the right satisfies
\begin{align}
(P_n-P)\left[\mathscr{L}_{r}^t(\cdot;\hat{\epsilon}_{r}^t)-\mathscr{L}_{r}^t(\cdot;0)\right]&=\begin{cases}
o_P(n^{-1/2}),&\mbox{ if }\norm{\hat{\epsilon}_{s}^{r}}=o_P(1), \\
O_P(n^{-1/2}),&\mbox{ otherwise.}
\end{cases}\label{eq:Donskercases}
\end{align}
A first application of the above result shows that $\norm{\hat{\epsilon}_{s}^{r}}=O_P(n^{-1/4})$. A second application shows that $\norm{\hat{\epsilon}_{s}^{r}}=o_P(n^{-1/4})$. Recall that $\norm{\hat{Q}_{s}^{r+1}-\hat{Q}_{s}^{r}}\lesssim \norm{\hat{\epsilon}_{s}^{r}}$, thereby establishing the first part of IH($r$). For the second part of IH($r$), note that it suffices to bound the two terms on the right-hand side of \eqref{eq:excessriskinduction}. The first term is controlled using that the right-hand side of \eqref{eq:riskub} is $o_P(n^{-1/2})$ under IH($r+1$). The second term is $o_P(n^{-1/2})$ by the above since $\norm{\hat{\epsilon}_{s}^{r}}=o_P(n^{-1/4})$. Hence, we have established the second part of IH($r$), namely that $\E[\mathcal{E}_{s}^{r}(\bar{H}_{r})]=o_P(n^{-1/2})$.
\end{proof}

\section{Right-censored data structures and time-to-event outcomes}\label{app:rtcens}
\begin{remark}[Right-censored data structures]\label{rem:rc}
General discretely right-censored data structures can be expressed using our notation. In what follows we mimic the introduction to discretely right-censored data structures given in \cite{Bang&Robins05}. For each $t$, let $\bar{L}_t\equiv (L_1,\ldots,L_t)$. Let $C$ be a discrete censoring time taking value in $1,\ldots,K+1$. The observation is censored after time $C$, so that we observe $(C,\bar{L}_C)$. Under the missing at random assumption, $P(C=t|C\ge t,\bar{L}_{K+1})=P(C=t|C\ge t,\bar{L}_t)$. If one wishes to estimate $\E[L_{K+1}|\bar{L}_t]$ for some $t\le K$, then one can use that, under the missing at random assumption, this estimand is equal to $\tilde{Q}_t(\bar{L}_t)$, where $\tilde{Q}_{s}$, $s\ge t$, is recursively defined as $\tilde{Q}_{K+1}(\bar{L}_{K+1})\equiv \bar{L}_{K+1}$, and $\tilde{Q}_{s}(\bar{L}_{s})\equiv \E[\tilde{Q}_{s+1}(\bar{L}_{s+1})|C> s,\bar{L}_{s}]$. To see the equivalence with our data structure, let $A_t\equiv \Ind_{\{C> t\}}$. Then $\E[L_{K+1}|\bar{\ell}_t]=\hat{Q}_t(\bar{h}_t)$, where $\bar{h}_t$ is the history vector with time $s\le t$ covariates $\ell_{s}$ and time $s<t$ treatment $A_{s}=1$.
\end{remark}

\begin{remark}[Time-to-event outcomes]\label{rem:tte}
Let $C$ be a censoring time taking values in $1,\ldots,K,+\infty$, let $T$ be a survival time taking values in $1,\ldots,K+1,+\infty$, and $\bar{L}_T\equiv (L_1,\ldots,L_T)$ denote a vector of covariates up to the survival time, where each $L_t$, $t\le K$, contains an indicator that $T\le t$ and $L_{K+1}=\Ind_{\{T\le K+1\}}$. We wish to estimate $P(T\le K+1|\bar{\ell}_t)$. One way to express the observed data structure is to write $(Y,\Delta,\bar{L}_{Y})$, where $Y\equiv \min\{T,C\}$ and $\Delta\equiv \Ind_{\{T\le C\}}$. Alternatively, the observed data structure can be expressed using our $\bar{H}_{K+1}$ notation, where each $A_t\equiv \Ind_{\{Y> t\}\cup\{\Delta=1\}}$. Under the sequential randomization assumption that $P(C=t|C\ge t,T> t,\bar{L}_T)=P(C=t|C\ge t,T> t,\bar{L}_t)$ with probability one, one can show that $P(T\le K+1|\bar{\ell}_t)$ is equal to $Q_t(\bar{h}_t)$, where each $a_{s}$, $s<t$, in $\bar{h}_t$ is equal to one. Working with time-to-event data requires one additional consideration compared to the longitudinal treatment setting that we consider in the main text. In particular, once the indicator that $T\le t$ in $\bar{L}_t$ is equal to one, it is automatically true that $L_{K+1}=1$. Thus, one should deterministically set estimates of $Q_t(\bar{L}_t)$ equal to one for all such $\bar{L}_t$.
\end{remark}

\section{Variation-independent formulation of sequential double robustness}\label{app:varindep}
We now present a variation independent formulation of sequential double robustness and establish its achievability. This formulation is more restrictive than that given in Definition~\ref{def:SDRQk}, but makes clear that there are scenarios in which one could correctly specify OR.$s$ but not OR.$t$, $s>t$. For each $t\ge 1$, let $P_t$ denote the distribution of $L_{t+1}$ conditional on $A_t=1,\bar{H}_t$ that is implied by $P$. Consider an estimation procedure that estimates each $P_t$ separately, where we note that these $P_t$ are variation independent, both of one another and of the treatment mechanisms, in the sense that knowing $P_t$ places no restriction on the set of possible realizations of $P_{s}$, $s\not=t$, or of $\pi_{s}$, $s$ arbitrary. Thus, so our procedure can estimate all of these conditional distributions \textit{a priori}. Define the following alternative to OR.$t$:
\begin{enumerate}
  \item[\mylabel{it:fd}{FD.$t$)}]  The distribution $P_t$ is correctly specified by the estimation procedure, or at least arbitrarily well approximated asymptotically.
\end{enumerate}
Above ``FD'' refers to correct specification of a component of the Full Data distribution. Consider an alternative definition of sequential double robustness that replaces OR.$t$ and OR.$s$ in Definition~\ref{def:SDRQk} by FD.$t$ and FD.$s$. First note that, by recursive applications of this definition, from $t=K,\ldots,0$, we see that a procedure satisfying this alternative SDR definition implicitly correctly specifies the functional form of the time $t$ outcome regression (i.e., satisfies OR.$t$) at each time point for which FD.$t$ holds and, for each $s>t$, either FD.$t$ or TM.$t$ holds. Secondly, an estimator achieving this alternative SDR property is achievable using the (S)DR unbiased transformation presented in Section~\ref{sec:existing}, where the regressions are fitted via kernel regression. Here we used the fact that kernel regression represents a kernel density estimation based plug-in estimator for the regression function. The downside to this procedure is that it requires correctly estimating possibly high-dimensional conditional densities. We therefore prefer an alternative approach that allows us to incorporate modern regression techniques -- see Section~\ref{sec:newproc}. Nonetheless, the variation independence of the procedure discussed in this appendix: FD.$t$ holding does not logically imply that FD.$s$ holds for $t\not=s$, just as TM.$t$ does not logically imply that TM.$s$ holds for $t\not=s$.

\section{Mitigating overfitting via cross-validation}\label{app:cv}
In this section, we describe a variant of $V$-fold cross-validation that can be used to estimate each $Q_t$. Let $\mathcal{S}^1,\ldots,\mathcal{S}^V$ be mutually exclusive and exhaustive subsets of $\{1,\ldots,n\}$ of approximately equal size, determined independently of the observations $\bar{H}_{K+1,1},\ldots,\bar{H}_{K+1,n}$. The set $\mathcal{S}^v$ is referred to as validation set $v$, and its complement is referred to as training set $v$. For each $i$, let $v_i$ denote the validation set $\mathcal{S}^{v_i}$ to which $i$ belongs, i.e. the validation set for observation $i$.

Suppose we wish to estimate some parameter $\theta$ of an arbitrary distribution $\nu$ on $Z\equiv (X,Y)$, e.g. $\theta(x)=\E_{\nu}[Y|x]$. Suppose that we have $M$ candidate regression algorithms for estimating $\theta$. For each candidate algorithm $m$ and subset indicator $v$, let $\hat{\theta}^{v,m}$ denote an estimate of $\theta$ by running algorithm $m$ on observations in training set $v$. The cross-validation selector for $m$ is defined as
\begin{align*}
\hat{m}&= \min_m \sum_{i=1}^n \mathscr{L}(Z;\hat{\theta}^{v_i,m})
\end{align*}
for some appropriately defined loss $\mathscr{L}$. If the loss function depends on nuisance parameters that must be estimated from the data, then one can replace this loss by a loss with the nuisance parameter estimates obtained from training set $v_i$, up to the dependence of these parameters on their own cross-validation selectors of an algorithm $m$. We let $\hat{\theta}^{v_i}\equiv \hat{\theta}^{v_i,\hat{m}}$. Note that $\hat{\theta}^{v_i}$ depends on the data in $v_i$ only through the selected $\hat{m}$ and, if the loss function depends on nuisance parameters, then also through the cross-validation-selected candidate algorithms from these nuisance parameters. As we will show in a future work, the fact that $\hat{\theta}^{v_i}$ only depends on validation set $i$ through discrete quantities ensures that our cross-validation scheme for estimating $\hat{m}$ satisfies oracle inequalities analogous to those presented in \cite{vanderLaan&Dudoit03}.

While the above procedure is written as a discrete cross-validation selector, we note that the super-learner algorithm, which replaces the discrete choice of size $M$ with all convex combinations of $M$ algorithms, can be arbitrarily well approximated by forming an $\epsilon$-net over the $M-1$ simplex, where now each convex combination of these $M$ algorithms is treated as a candidate \citep{vanderLaan&Dudoit&vanderVaart06}.

\begin{algorithm}
\caption{Cross-validated iTMLE for $Q_t$}\label{alg:cvsdr}
\begin{algorithmic}[1]
\Statex \parbox[t]{\dimexpr\linewidth-\algorithmicindent}{Selects between $M_t$ with covariate $\bar{H}_t$, $t\ge t^0$. When $t^0=0$, $M_{t^0}=1$ and the only candidate regression is an intercept-only logistic regression.}
\Procedure{SDR.Q}{$t^0$}
\For{v=1,\ldots,V} \Comment{Estimate treatment mechanisms.}
  \State \parbox[t]{\dimexpr\linewidth-\algorithmicindent-\algorithmicindent}{Using only observations in training set $v$, obtain estimates $\hat{\pi}_{t,v}$ of $\pi_t$, $t=t^0,\ldots,K$, via any desired technique.}
\EndFor
\State Let $\hat{Q}_{K+1}^*\equiv L_{K+1}$
\For{$s=K,\ldots,t^0$} \Comment{Estimate $Q_{s}$}
  \State Initialize $\hat{Q}_{s}^{s+1}\equiv 1/2$
  \For{$t=s,s-1,\ldots,t^0$} \Comment{Target estimate of $Q_{s}$ towards $Q_{t}$.}
    \For{v=1,\ldots,V} \Comment{Fit candidates on training set $v$.}
      \State For each function $\epsilon_{s}^{t} : \bar{\mathcal{H}}_{t}\rightarrow\mathbb{R}$, define $\hat{Q}_{s}^{t+1,\epsilon_{s}^{t},v}$ as
      \begin{align*}
      \bar{h}_{t}\mapsto \Psi\left\{\Psi^{-1}[\hat{Q}_{s}^{t+1,v}(\bar{h}_{s})] + \epsilon_{s}^{t}(\bar{h}_{t})\right\}.
      \end{align*}
      \For{m=1,\ldots,M} \Comment{Fit the $m^{\textnormal{th}}$ candidate}
        \State \parbox[t]{\dimexpr\linewidth-\algorithmicindent-\algorithmicindent-\algorithmicindent-\algorithmicindent-\algorithmicindent}{Using observations $i$ in training set $v$, fit $\hat{\epsilon}_{s}^{t,v,m}$ by running candidate regression $m$ using the cross-entropy loss and the logit link function with outcome $\hat{Q}_{s+1,i}^{*,v}$, offset $\logit \hat{Q}_{s,i}^{t+1,v}$, predictor $\bar{H}_{t,i}$, and weights $A_{t,i}\prod_{r=t+1}^{s} \frac{A_{r,i}}{\hat{\pi}_{r,v,i}}$.
        }
      \EndFor
    \EndFor
    \State \parbox[t]{\dimexpr\linewidth-\algorithmicindent}{Define
    \begin{align*}
    \hat{m}_{s}^{t}&= \argmin_m \sum_{i=1}^n \mathscr{L}_{s}^{t,v_i}(\bar{H}_{t+1,i};\hat{\epsilon}_{s}^{t,v_i,m}),
    \end{align*}
    where $\mathscr{L}_{s}^{t,v}(\bar{h}_{s+1};\epsilon)$ is defined as
    \begin{align*}
    &-a_t\left\{\prod_{r=t+1}^{s} \frac{a_{r}}{\hat{\pi}_{r}^v}\right\}\left\{\hat{Q}_{s+1}^{*,v}\log \hat{Q}_{s}^{t+1,\epsilon,v}  + [1-\hat{Q}_{s+1}^{*,v}]\log\left[1-\hat{Q}_{s}^{t+1,\epsilon,v}\right]\right\}
    \end{align*}
    }
    \State Let $\hat{Q}_{s}^{t,v}\equiv \hat{Q}_{s}^{t+1,\hat{\epsilon}_{s}^{t,v,\hat{m}_{s}^t},v}$, $v=1,\ldots,V$
  \EndFor
  \State Let $\hat{Q}_{s}^{*,v}\equiv \hat{Q}_{s}^{t^0,v}$, $v=1,\ldots,V$ \Comment{Targeted towards all $Q_{t}$, $t<s$}
\EndFor
\State \parbox[t]{\dimexpr\linewidth-\algorithmicindent}{Using observations $i=1,\ldots,n$, let $\hat{Q}_{t^0}^*$ be the output of candidate regression $\hat{m}_{t^0}^{t^0}$ run with outcome $\hat{Q}_{t^0+1,i}^{*,v_i}$, predictor $\bar{H}_{t^0,i}$, no offset, and weights $A_{t^0,i}$.}
\State \textbf{return} $\hat{Q}_{t^0}^*$
\EndProcedure
\end{algorithmic}
\end{algorithm}

\section{Simulated data-generating distributions}\label{app:simDGD}
All simulations are carried out in R programming
package using $\mathtt{simcausal}$ package \citep{Sofryginetal2015}.

\subsection{Simulation 1}\label{app:simDGD.sim1}
The simulation is implemented by sampling longitudinal data over $3$ time-points
and $n=500$ i.i.d. subjects. Briefly, this simulation represents a simple data
structure with binary exposure that can be assigned separately at time-points
$t=1,\ldots,3$ and a binary outcome of interest $Y_{K}$ evaluated at
$K=3$.

Recall from the main text that
$\Psi(x)\equiv1/(1+e^{-x})$. The longitudinal 
structure on each subject was sampled according to the following structural 
equation model for time-point $t=1$:

{\singlespacing\begin{align*}
&L_{t}\sim\Normal(0,1)\\
&A_{t}\sim\Bern\left(\Psi\{L_{t}\}\right)\\
&Y_{t}=0. \\
\intertext{ Followed by time-point $t=2$:}
&L_{t}\sim\Normal(0,1)\\
&A_{t}\sim\Bern\left(\Psi\{L_{t}+A_{t-1}\}\right)\\
&Y_{t}=0. \\
\intertext{Followed by time-point $t=3$:}
&L_{t}\sim\Normal(L_{t-2} A_{t-1}+A_{t-2} L_{t-1}+L_{t-1} A_{t-1},1)\\
&A_{t}\sim\Bern\left(\Psi\{L_{t}+A_{t-1}\}\right)\\
&Y_{t}=\Bern\left(\Psi\{L_{t-1} A_{t}+A_{t-1} L_{t}+L_{t} A_{t}\}\right).
\end{align*}}

\subsection{Simulation 2}\label{app:simDGD.sim2}
The simulation is implemented by sampling longitudinal data over $5$ time-points
and $n=5,000$ i.i.d. subjects. The data-generating distribution for Simulation 2 is more complex and higher dimensional 
than that in Simulation 1. Briefly, for Simulation 2 we let $A_{t}=1$
denote standard of care and $A_{t}=0$ denote the experimental new treatment
at time-point $t$ (note: our coding is the reverse of the more standard
way to denote $A_{t}=1$ as the new treatment at $t$). The subject can
switch from the standard of care at any time during the follow-up $t=1,\ldots,K$. However, once the subject switches he or she is forced to stay on the new treatment until the end of the follow-up.
We assume that the outcome of interest, $Y_{K}$, is a binary indicator
of the adverse event at the final follow-up time-point $K$. Furthermore, switching to the experimental treatment $A_{t}=0$ at any $t$ lowers the probability of the final adverse event at $K$. Finally, the probability of receiving the experimental treatment $A_{t}=0$ 
increases once the subject's risk for experiencing the adverse
end-of-the study event becomes high, where subject's risk at each $t$
is being assessed conditionally on the fact that he or she will remain on the
standard of care. Thus, an incorrectly specified regression of $\pi_{t}$ in such a data-generating process would miss the informative
switching to the experimental treatment, yielding a biased estimate of $Q_{0}$ for IPW-based estimator that requires consistent
$\hat{\pi}_{t}$.

The longitudinal structure on each subject is sampled according
to the following structural equation model at time-point $t=1$:

{\singlespacing\begin{align*}
&U_{L,t}\sim\Normal(0,1)\\
&W_{t}\sim\Normal(0,1)\\
&L_{t}=\left|U_{L,t}\right|\\
&\phR_{t}=L_{t}\\
&A_{t}\sim\Bern\left(\Psi\{L_{t}\}\right)\\
&Y_{t}=0. \\
\intertext{Followed by time-point $t=2$:}
&U_{L,t}\sim\Normal(0,1)\\
&L_{t}=\left|U_{L,t}\right|\\
&\phR_{t}=-2+0.5 L_{t-1}+L_{t}\\
&A_{t}\sim\Bern\left(A_{t-1} \Psi\{1.7-2.0 \Ind_{\left\{\Psi(\phR_{t})>0.9\right\}}\}\right)\\
&Y_{t}\sim\Bern\left(\Psi\{-3+0.5 L_{t-1} A_{t}+0.5 A_{t-1} L_{t}+0.5 L_{t} A_{t}\}\right).\\
\intertext{Followed by time-point $t=3$:}
&U_{L,t}\sim\Normal(A_{t-2} L_{t-1}+L_{t-1} A_{t-1},1)\\
&L_{t}=\left|U_{L,t}\right|\\
&\phR_{t}=-2+0.5 L_{t-1}+L_{t}\\
&A_{t}\sim\Bern\left(A_{t-1} \Psi\{1.7-2.0 \Ind_{\left\{\Psi(\phR_{t})>0.85\right\}}\}\right)\\
&Y_{t}\sim\Bern\left(\Psi\{-3 Y_{t-1}+0.5 L_{t-1} A_{t}+0.5 A_{t-1} L_{t}+0.5 L_{t} A_{t}\}\right).\\
\intertext{Followed by time-point $t=4$:}
&U_{L,t}\sim\Normal(L_{t-2} A_{t-1}+A_{t-2} L_{t-1}+L_{t-1} A_{t-1},1)\\
&L_{t}=\left|U_{L,t}\right|\\
&\phR_{t}=-2+0.5 L_{t-1}+L_{t}\\
&A_{t}\sim\Bern\left(A_{t-1} \Psi\{1.7-2.0 \Ind_{\left\{\Psi(\phR_{t})>0.80\right\}}\}\right)\\
&Y_{t}\sim\Bern\left(\Psi\{-3 Y_{t-1}+0.5 L_{t-1} A_{t}+0.5 A_{t-1} L_{t}+0.5 L_{t} A_{t}\}\right).\\
\intertext{Finally, followed by time-point $t=5$:}
&U_{L,t}\sim\Normal(L_{t-3} A_{t-1}+A_{t-3} L_{t-1}+L_{t-1} A_{t-1},1)\\
&L_{t}=\left|U_{L,t}\right|\\
&\phR_{t}=-1+0.25 L_{t-1}+0.5 L_{t}-0.1 L_{t} L_{t-1}+1.5 W_{0} L_{t-1}\\
&A_{t}\sim\Bern\left(A_{t-1} \Psi\{2-2.0 \Ind_{\left\{\Psi(\phR_{t})>0.80\right\}}\}\right)\\
&Y_{t}\sim\Bern\left(\Psi\{-Y_{t-1}+A_{t}+\phR_{t} A_{t}+0.20 A_{t-1} L_{t}\}\right).
\end{align*}}

\section{Further Algorithms}

\begin{algorithm}
\caption{SDR estimation of $Q_t$ via doubly robust transformations \citep[particular implementation of][]{Rubin&vanderLaan2007}}\label{alg:druntrnas}
\begin{algorithmic}[1]
\Statex \parbox[t]{\dimexpr\linewidth-\algorithmicindent}{This function runs regressions using covariate $\bar{H}_t$ for each $t$, where the regression algorithms may be $t$-dependent.}
\Procedure{SDR.Unbiased.Trans}{}
\State Let $\hat{Q}_{K+1}(\bar{H}_{K+1})\equiv L_{K+1}$.
\State \parbox[t]{\dimexpr\linewidth-\algorithmicindent}{Obtain estimates $\hat{\pi}_t$, $t=1,\ldots,K$, via any desired technique.}
\For{$t=K,\ldots,0$}
\State \parbox[t]{\dimexpr\linewidth-\algorithmicindent}{Using observations $i$ with $A_{t,i}=1$, regress $\Gamma_{t,i}$ against $\bar{H}_{t,i}$, where
\begin{align*}
\widehat{\Gamma}_{t,i}&\equiv \sum_{s=t+1}^K \left(\prod_{r=t+1}^{s} \frac{A_{r,i}}{\hat{\pi}_{r,i}}\right)\left\{\hat{Q}_{s+1,i} - \hat{Q}_{s,i}\right\} + \hat{Q}_{t+1,i}.
\end{align*}
Label the output $\hat{Q}_t$.}
\EndFor
\State \textbf{return} $\{\hat{Q}_t : t=0,\ldots,K\}$
\EndProcedure
\end{algorithmic}
\end{algorithm}

\begin{algorithm}
\caption{DR estimation of $Q_0$ \citep[variant of][]{Bang&Robins05}}\label{alg:bangrobins}
\begin{algorithmic}[1]
\Statex \parbox[t]{\dimexpr\linewidth-\algorithmicindent}{For each $t\ge 1$, $g_t^{\beta_t} : \bar{\mathcal{H}}_t\rightarrow\mathbb{R}$ is a parametric model for $Q_t$ indexed by $\beta_t\in\mathbb{R}^{d_t}$, where $d_t$ is finite. Let $\beta_0=\emptyset$, and $g_0^{\beta_0}\equiv 1/2$.}
\Procedure{DR.Q}{}
\State 	Let $\hat{Q}_{K+1}(\bar{H}_{K+1})\equiv L_{K+1}$.
\State \parbox[t]{\dimexpr\linewidth-\algorithmicindent}{Obtain estimates $\hat{\pi}_t$, $t=1,\ldots,K$, via any desired technique.}
\For{$t=K,\ldots,0$}
\State For $(\beta_t,\epsilon_t)\in \mathbb{R}^{d_t+1}$, define $\hat{Q}_t^{\beta_t,\epsilon_t}$ as
\begin{align*}
\bar{h}_t\mapsto \Psi\left\{g_t^{\beta_t}(\bar{h}_t) + \frac{\epsilon_t}{\prod_{s=1}^{t}\hat{\pi}_{s}(\bar{h}_{s})}\right\}.
\end{align*}
\State \parbox[t]{\dimexpr\linewidth-\algorithmicindent}{Define $(\hat{\beta}_t,\hat{\epsilon}_t)$ as the arguments $(\beta_t,\epsilon_t)$ minimizing
\begin{align*}
-\sum_{i : \prod_{s=1}^t A_{s,i}=1} \Bigg[&\hat{Q}_{t+1,i} \log \hat{Q}_{t,i}^{\beta_t,\epsilon_t} + \{1-\hat{Q}_{t+1,i}\}\log\{1-\hat{Q}_{t,i}^{\beta_t,\epsilon_t}\}\Bigg].
\end{align*}
}
\LineCommentCont{If $g_t^{\beta_t}(\bar{h}_t)=\langle \textnormal{R}(\bar{h}_t),\beta_t\rangle$ for some transformation $\textnormal{R}$, then the above can be optimized in the programming language \texttt{R} by running a linear-logistic regression of the $[0,1]$-bounded $\hat{Q}_{t+1,i}$ against $\textnormal{R}(\bar{H}_{t,i})$ and $1/\prod_{s=1}^{t}\hat{\pi}_{s}(\bar{h}_{s})$ among all individuals with $\prod_{s=1}^t A_{s,i}=1$.}
\State Let $\hat{Q}_t\equiv \hat{Q}_t^{\hat{\beta}_t,\hat{\epsilon}_t}$.
\EndFor
\State \textbf{return} $\{\hat{Q}_t : t=0,\ldots,K\}$
\EndProcedure
\end{algorithmic}
\end{algorithm}

\begin{algorithm}
\caption{iTMLE for $Q_t$}\label{alg:sdr}
\begin{algorithmic}[1]
\Statex \parbox[t]{\dimexpr\linewidth-\algorithmicindent}{This function runs regressions using covariate $\bar{H}_t$ for each $t\ge t^0$. These regressions may be $t$-dependent. In particular, they should be intercept-only logistic regressions if $t=0$.}
\Procedure{SDR.Q}{$t^0$}
\State \parbox[t]{\dimexpr\linewidth-\algorithmicindent}{Obtain estimates $\hat{\pi}_t$, $t=t^0+1,\ldots,K$, via any desired technique.}
\State Let $\hat{Q}_{K+1}^*\equiv L_{K+1}$
\For{$s=K,\ldots,t^0$} \Comment{Estimate $Q_{s}$}
	\State Initialize $\hat{Q}_{s}^{s+1}\equiv 1/2$
	\For{$t=s,s-1,\ldots,t^0$} \Comment{Target estimate of $Q_{s}$ towards $Q_{t}$.}
		\State For each function $\epsilon_{s}^{t} : \bar{\mathcal{H}}_{t}\rightarrow\mathbb{R}$, define $\hat{Q}_{s}^{t+1,\epsilon_{s}^{t}}$ as
		\begin{align*}
		\bar{h}_{s}\mapsto \Psi\left\{\Psi^{-1}[\hat{Q}_{s}^{t+1}(\bar{h}_{s})] + \epsilon_{s}^{t}(\bar{h}_{t})\right\}.
		\end{align*}
		\State \parbox[t]{\dimexpr\linewidth-\algorithmicindent-\algorithmicindent-\algorithmicindent}{Using all observations $i=1,\ldots,n$, fit $\hat{\epsilon}_{s}^{t}$ by running a regression using the cross-entropy loss and the logit link function with:
		\begin{itemize}
			\item Outcome: $\hat{Q}_{s+1}^*$
			\item Offset: $\logit \hat{Q}_{s}^{t+1}$
			\item Predictor: $\bar{H}_{t}$
			\item Weight: $A_{t}\prod_{u=t+1}^{s} \frac{A_{u}}{\hat{\pi}_{u}}$
		\end{itemize}
		}
		\State $\hat{Q}_{s}^{t}\equiv \hat{Q}_{s}^{t+1,\hat{\epsilon}_{s}^{t}}$
	\EndFor
	\State Let $\hat{Q}_{s}^*\equiv \hat{Q}_{s}^{t^0}$ \Comment{Targeted towards all $Q_{t}$, $t<s$}
\EndFor
\State \textbf{return} $\hat{Q}_{t^0}^*$
\EndProcedure
\end{algorithmic}
\end{algorithm}

\begin{figure}[H]
\begin{centering} \includegraphics[width=.75\linewidth]{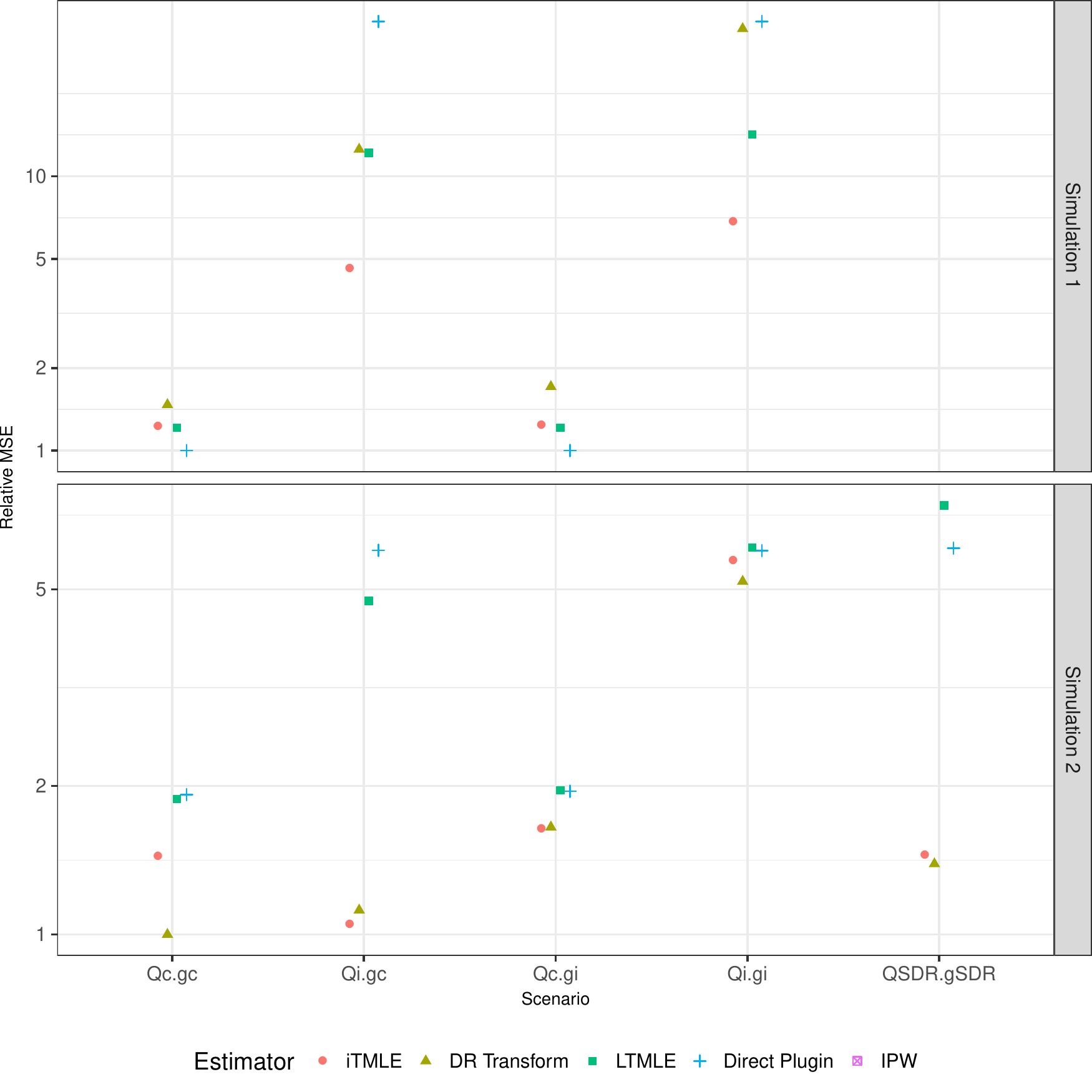}
\par\end{centering}

\caption[Relative MSE for $\hat{Q}_1$ for simulation scenario 1 (top panel) and simulation scenario 2 (bottom panel)]{Relative MSE for $\hat{Q}_1$ for simulation scenario 1 (top panel) and simulation scenario 2 (bottom panel). Simulation 1 is based on longitudinal data with 3 time-points and $n$=500 observations. Simulation 2 is based on longitudinal data with 5 time-points and $n$=5,000 observations. The iTMLE and DR Transform typically outperform or perform comparably to both competitors.}\label{fig:simres.all.MSE} 
\end{figure}

\begin{figure}[H]
\begin{centering} \includegraphics[width=.75\linewidth]{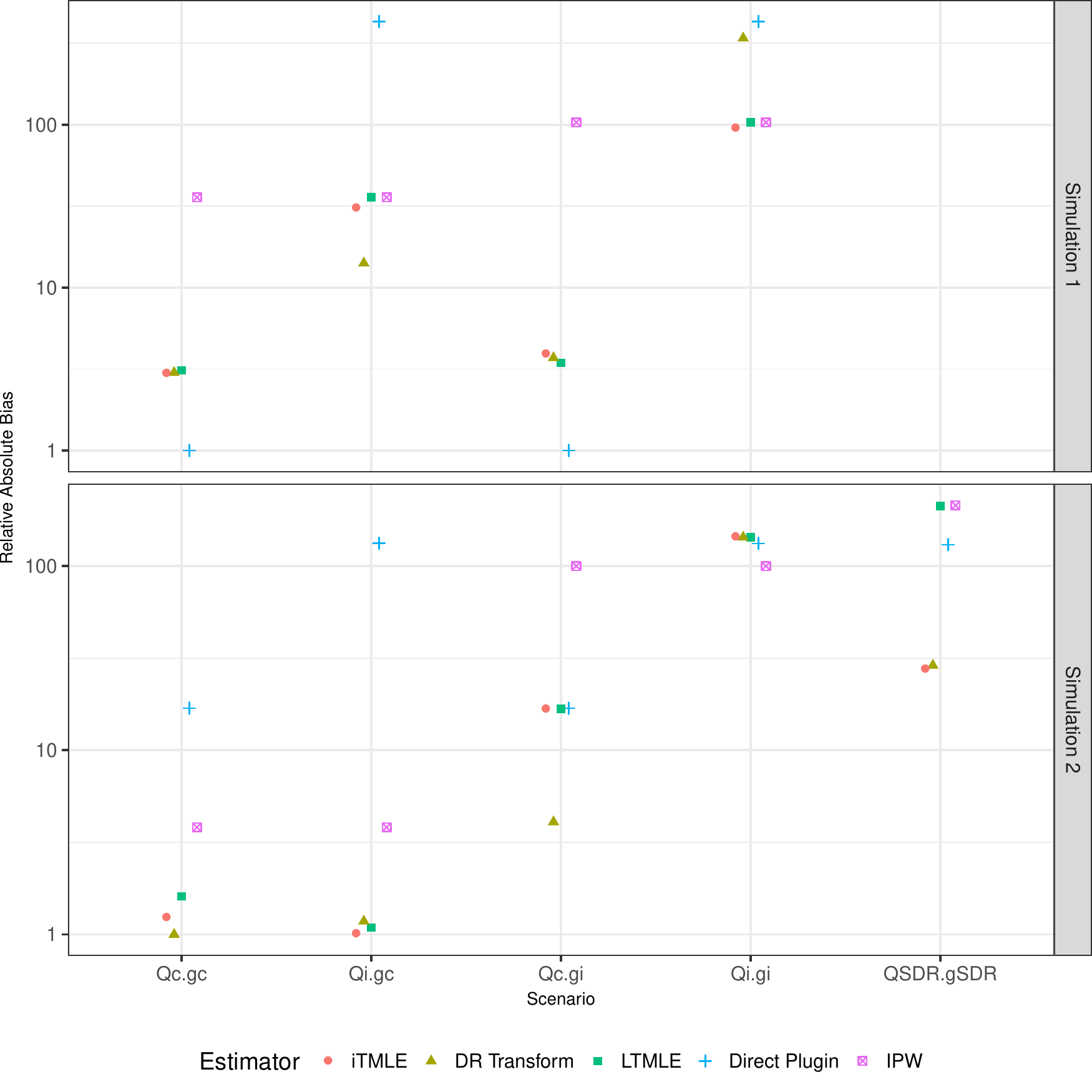}
\par\end{centering}

\caption{Relative absolute bias for $\hat{Q}_0$ for simulation scenario 1 (top panel) and simulation scenario 2 (bottom panel). Simulation 1 is based on longitudinal data with 3 time-points and $n$=500 observations. Simulation 2 is based on longitudinal data with 5 time-points and $n$=5,000 observations. The performance of LTMLE, iTMLE, and DR Transform is similar. The only exception for Simulation 1 is under \textit{Qi.gc}, where DR Transform outperforms other methods. The only exceptions for Simulation 2 are for \textit{Qc.gi}, where DR Transform outperforms other methods, and \textit{QSDR.gSDR}, where both SDR methods outperform LTMLE.}\label{fig:simres.all.BIAS}
\end{figure}

\begin{figure}[H]
\begin{centering} \includegraphics[width=.9\linewidth]{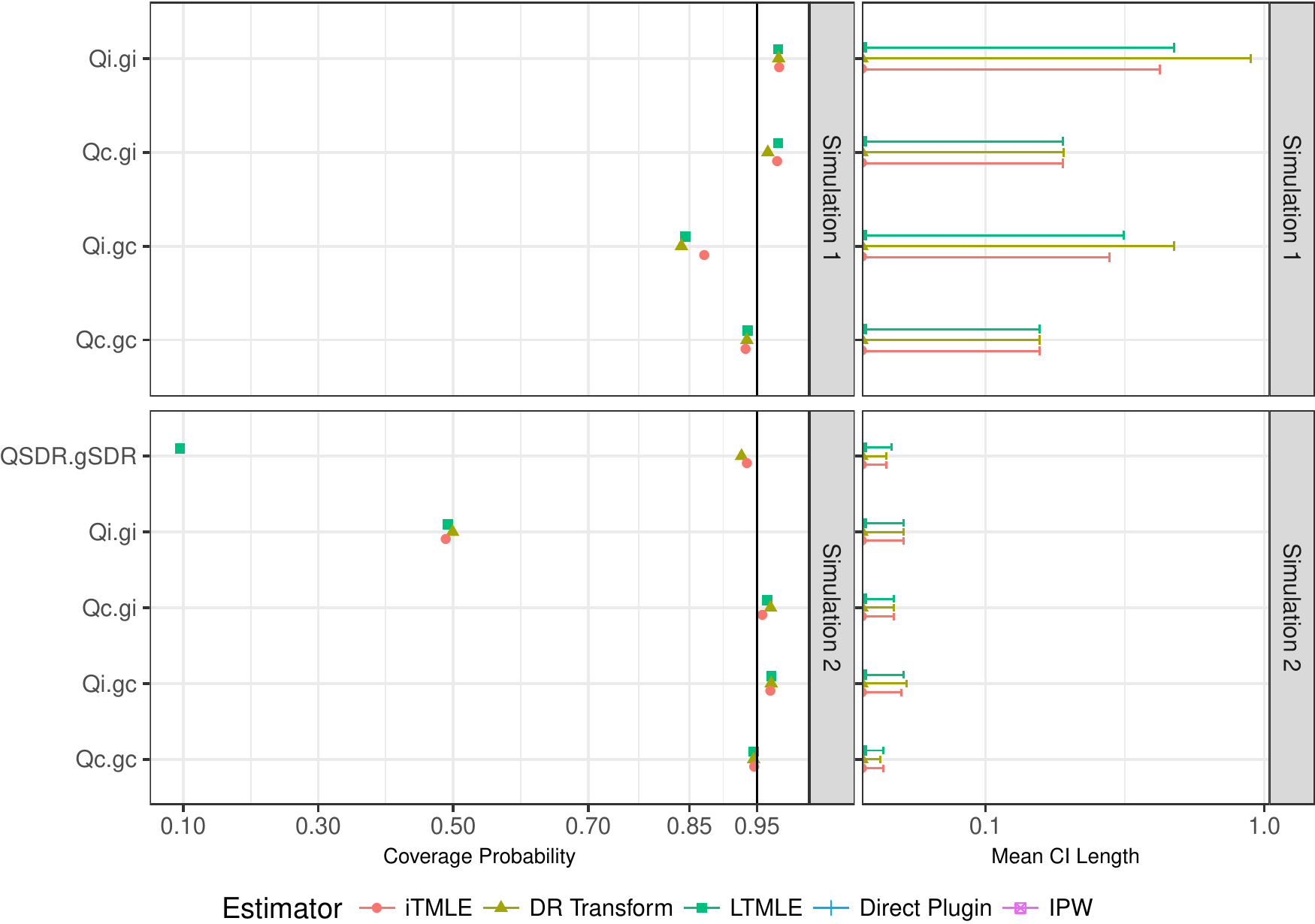}
\par\end{centering}

\caption{Coverage (left panels) and mean length (right panels) of the two-sided 95\% CIs for $Q_0$ in simulation scenario 1 (top panels) and simulation scenario 2 (bottom panels). Confidence interval coverage and width appear to be comparable between the two SDR methods and the LTMLE. The only exception is for the \textit{QSDR.gSDR} scenario, where the LTMLE has roughly 10\% coverage, whereas the SDR  approaches nearly achieve the nominal coverage level.}\label{fig:simres.all.coverCIlen}
\end{figure}

\end{document}